\documentclass[authoryear,preprint,10pt]{elsarticle}

\usepackage[USenglish]{babel}

\usepackage{amssymb, dsfont, mathtools}
	\usepackage{eurosym}
	\usepackage{color}
    \usepackage{comment}
	\usepackage{bm}
	\usepackage{setspace}
	\usepackage{booktabs}
	\usepackage{tabularx,multirow}
	\usepackage{longtable}
	\usepackage{arydshln}
	\usepackage{paralist}
	\usepackage{graphicx}
	\usepackage{adjustbox}
	\usepackage{afterpage}
	\usepackage{pdflscape}
	\usepackage{lscape}
	\usepackage{float}
	\usepackage[caption = false]{subfig}
	\usepackage{tikz, pgfplots}
	\usepackage{tikz}
    \usetikzlibrary{shapes.geometric, arrows}
    \usepackage{pgfplots}
	\usepgfplotslibrary{fillbetween}
	\usepackage[inline]{enumitem}
	\usetikzlibrary{arrows}
	\usepackage{lineno, blindtext}
	\usepackage{forest}
	\usepackage{rotating}
	\usepackage{tikz}
    \usepackage{pgfplotstable}
	\usetikzlibrary{shapes,shadows,arrows,positioning,mindmap}
	\pgfdeclarelayer{background}
	\pgfdeclarelayer{foreground}
	\pgfsetlayers{background,main,foreground}

	\usepackage{amsthm,amsmath,eurosym}
	
	\newlist{inlinelist}{enumerate*}{1}
	\setlist*[inlinelist,1]{%
		label=(\roman*),
	}
	\newcommand{\subscript}[2]{$#1 _ #2$}

\usepackage{algorithm}
	\usepackage{algpseudocode}
	\makeatletter	
	\def\algbackskip{\hskip-\ALG@thistlm}	
	\makeatother
	
	\usepackage[acronym]{glossaries}
	\newglossaryentry{DEA}
	{
		name={DEA},
		description={Data Envelopment Analysis},
		first={\glsentrydesc{DEA} (\glsentrytext{DEA})},
		plural={DEA},
		firstplural={\glsentrydesc{DEA} (\glsentryplural{DEA})}
	}
	\newglossaryentry{DMU}
	{
		name={DMU},
		description={Decision Making Unit},
		first={\glsentrydesc{DMU} (\glsentrytext{DMU})},
		plural={DMUs},
		firstplural={\glsentrydesc{DMU}s (\glsentryplural{DMU})}
	}
    \newglossaryentry{HR}
	{
		name={HR},
        description={Hit \& Run},
		first={\glsentrydesc{HR} (\glsentrytext{HR})},
	}
    \newglossaryentry{CRS}
	{
		name={CRS},
        description={Constant Returns to Scale},
		first={\glsentrydesc{CRS} (\glsentrytext{CRS})},
	}
	\newglossaryentry{VRS}
	{
		name={VRS},
        description={Variable Returns to Scale},
		first={\glsentrydesc{VRS} (\glsentrytext{VRS})},
	}
    \newglossaryentry{PCA}
	{
		name={PCA},
        description={Principal Component Analysis},
		first={\glsentrydesc{PCA} (\glsentrytext{PCA})},
	}
    \newglossaryentry{MCAR}
	{
		name={MCAR},
        description={Missing Completely at Random},
		first={\glsentrydesc{MCAR} (\glsentrytext{MCAR})},
	}
    \newglossaryentry{MAR}
	{
		name={MAR},
        description={Missing at Random},
		first={\glsentrydesc{MAR} (\glsentrytext{MAR})},
	}
    \newglossaryentry{MNAR}
	{
		name={MNAR},
        description={Missing Not at Random},
		first={\glsentrydesc{MNAR} (\glsentrytext{MNAR})},
	}
    \newglossaryentry{IKCAR}
	{
		name={IKCAR},
        description={Imperfectly Known Completely at Random},
		first={\glsentrydesc{IKCAR} (\glsentrytext{IKCAR})},
	}
    \newglossaryentry{IKAR}
	{
		name={IKAR},
        description={Imperfectly Known at Random},
		first={\glsentrydesc{IKAR} (\glsentrytext{IKAR})},
	}
    \newglossaryentry{IKNAR}
	{
		name={IKNAR},
        description={Imperfectly Known Not at Random},
		first={\glsentrydesc{IKNAR} (\glsentrytext{IKNAR})},
	}
     \newglossaryentry{SMAA}
	{
		name={SMAA},
        description={Stochastic Multicriteria Acceptability Analysis},
		first={\glsentrydesc{SMAA} (\glsentrytext{SMAA})},
	}
	\newglossaryentry{IKD}
	{
		name={IKD},
        description={imperfect knowledge of data},
		first={\glsentrydesc{IKD} (\glsentrytext{IKD})},
	}

\usepackage{amsthm,amsmath,amssymb,eurosym}

\newtheorem{definition}{Definition}
\newtheorem{remark}{Remark}
\newtheorem{proposition}{Proposition}

\newtheorem{example}{Example}

\usepackage{lineno, blindtext}
\usepackage[
breaklinks=true,colorlinks=true,
linkcolor=blue,urlcolor=blue,citecolor=blue]{hyperref}
\pgfplotsset{compat=1.14}

\usepackage[scr=boondoxo,scrscaled=1.05]{mathalfa}

\numberwithin{definition}{section} 
\numberwithin{remark}{section} 
\numberwithin{example}{section} 
\numberwithin{proposition}{section} 
\numberwithin{equation}{section}
\numberwithin{table}{section} 
\numberwithin{figure}{section}

\usepackage{etoolbox}
\makeatletter
\patchcmd{\ps@pprintTitle}{\footnotesize\itshape
Preprint submitted to \ifx\@journal\@empty Elsevier
\else\@journal\fi\hfill\today}{\relax}{}{}
\makeatother

\usepackage{setspace}

\makeatletter
\def\ps@pprintTitle{%
   \let\@oddhead\@empty
   \let\@evenhead\@empty
   \let\@oddfoot\@empty
   \let\@evenfoot\@oddfoot
}
\makeatother

\begin{document}

\begin{frontmatter}
    \title{Data Envelopment Analysis models with imperfect knowledge of\\ input and output values: An application to Portuguese public hospitals}
    \author[ceris]{{\sc{D.C. Ferreira}}\corref{correspondingauthor}}
	\cortext[correspondingauthor]{Corresponding author. Address: CERIS, Instituto Superior T\'{e}cnico,  Universidade de Lisboa, Lisbon, Portugal, Avenida Rovisco Pais, 1, 1049-001, Lisboa, Portugal, Phone: +351 218 418 394, Fax: +351 218 497 650}
	\ead{diogo.cunha.ferreira@tecnico.ulisboa.pt}
	
	\author[ceg-ist]{{\sc{J.R. Figueira}}}
	\ead{figueira@tecnico.ulisboa.pt}
	
	\author[catania,uk]{{\sc{S. Greco}}}
	\ead{salgreco@unict.it, Salvatore.Greco@port.ac.uk}
	
	\author[ceris]{{\sc{R.C. Marques}}}
	\ead{rui.marques@tecnico.ulisboa.pt}

	\address[ceris]{CERIS, Instituto Superior T\'{e}cnico,  Universidade de Lisboa, Lisbon, Portugal}
	\address[ceg-ist]{CEG-IST, Instituto Superior T\'{e}cnico,  Universidade de Lisboa, Lisbon, Portugal}
	\address[catania]{Department of Economics and Business, University of Catania, Catania, Italy}
	\address[uk]{Portsmouth Business School, Operations \& Systems Management \\ University of Portsmouth,  Portsmouth, United Kingdom}

    \begin{abstract}
    \noindent  Assessing the technical efficiency of a set of observations requires that the associated data composed of inputs and outputs are perfectly known. If this is not the case, then biased estimates will likely be obtained. Data Envelopment Analysis (DEA) is one of the most extensively used mathematical models to estimate efficiency. It constructs a piecewise linear frontier against which all observations are compared. Since the frontier is empirically defined, any deviation resulting from low data quality (imperfect knowledge of data or IKD) may lead to efficiency under/overestimation. In this study, we model IKD and, then, apply the so-called Hit \& Run procedure to randomly generate admissible observations, following some prespecified probability density functions. Sets used to model IKD limit the domain of data associated with each observation. Any point belonging to that domain is a candidate to figure out as the observation for efficiency assessment. Hence, this sampling procedure must run a sizable number of times (infinite, in theory) in such a way that it populates the whole sets. The DEA technique is used during the execution of each iteration to estimate bootstrapped efficiency scores for each observation. We use some scenarios to show that the proposed routine can outperform some of the available alternatives. We also explain how efficiency estimations can be used for statistical inference. An empirical case study based on the Portuguese public hospitals database (2013-2016) was addressed using the proposed method.
    \end{abstract}
    \vspace{0.25cm}
\end{frontmatter}

\vfill\newpage

\section{Introduction}\label{sec:Introduction}

\noindent Assessing the efficiency of a set of observations, from now on called \glspl{DMU}, is an economic concern of any field. Usually, we are interested in assessing whether companies, either public or private, can reduce the resources wasted (keeping the delivered services or produced goods), raise their production levels (consumed resources held), or both. For instance, several authors have identified considerable inefficiency in public health care provision translating into many resources that could be saved if providers would be efficient \citep{Ferreira2015, FERREIRA2018b}. This problem exacerbates in times of pandemic outbreaks that increase the demand for such services \citep{RAMANATHAN2020518}. 

Unfortunately, measuring each \gls{DMU}'s accurate level of efficiency is impossible unless we know perfectly the function describing the production/consumption profile of each group of homogeneous \glspl{DMU}. Since it is often unavailable in the real world, we can only guess (estimate) the efficiency using the practical information of the group of \glspl{DMU}. Traditionally, it results from using a frontier against which the \glspl{DMU} are compared. The efficiency of a \gls{DMU} is, roughly speaking, the relationship between its input-output observed values and targets \textit{fixed} by the frontier. The problem, then, lies in the construction of the frontier, which can be based on a predefined parametric function or empirical data \cite{Daraio2007}. \gls{DEA} \citep{Charnes79, BANKER1984} is one of the most widely employed models to estimate efficiency, especially in the public sector including health care \citep{Hollings2008, Cordero2013, Ferreira2017b}. \gls{DEA} is a non-parametric technique as it is only empirically-based. \gls{DEA}'s fundamental strength is that it works based on data only, without making any assumption on the relationships mentioned above. The estimated frontier is composed of the efficient \glspl{DMU} and all possible linear convex combinations of them \citep{Cordero2016}. Thus, the efficient frontier is a continuum of input-output vectors. Targets for inefficient \glspl{DMU} result from the linear combination of observations associated with the efficient ones.

Nonetheless, empirical research in social sciences is often plagued by the \gls{IKD}, resulting from: uncertainty, imprecision, ill-determinations, arbitrariness, and missing data \citep{Hayek1945, French1995, RoyEtAl14}. \gls{IKD} can be either epistemic or aleatory, resulting from limited knowledge or randomness/variability in data, respectively. Meanwhile, these imperfect knowledge sources can be due to human errors, experimental failures, and disclosure restrictions, to name just a few. Models used to estimate efficiency, including \gls{DEA}, are typically sensitive to data quality \citep{Witte2010}. Hence, \gls{IKD} is a serious problem to get reliable or robust efficiency estimates for the \glspl{DMU}. A \gls{DMU} is \textit{robust efficient} whenever it remains efficient for all the input-output observations, regardless of the data quality. It is \textit{sufficiently robust efficient} when it is not efficient for all the input-output observations but it can be part of the frontier in a given number of times (threshold). 

Several approaches have been proposed in the literature to account for and dealing with the problem of \gls{IKD}. The way of modeling such imperfect knowledge includes, for instance, the omission of data, their substitution, and the use of feasible possible sets of values. The approaches comprise \gls{DEA} standard techniques, stochastic programming, fuzzy-possibilistic programming, interval programming, and robust optimization \citep{DESPOTIS200224,Kao2007,Kuosmanen2009_blank,EHRGOTT2018231}. Although the alternatives mentioned above have their merits, they also exhibit some disadvantages and caveats. On the one hand, some alternatives are difficult to implement, requiring a considerable computational effort (since they are of exact algorithmic nature). On the other hand, others sometimes impose objectionable substitutions of imperfectly known data or disregard the modeling of \gls{IKD}, allowing neither the analysis of efficiency distributions nor statistical inference. For instance, some do not allow us to classify a \gls{DMU} as robust or sufficiently robust efficient, or as  perfectly robust efficient, potentially robust efficient and robust inefficient units as in \cite{WEI201728}. 

This paper proposes an alternative based on the so-called \gls{HR} procedure \citep{Smith1984,Belisle1993,Kaufman1998}, which avoids the previous shortcomings. This alternative lies in defining of a set (and its boundary) in the space containing all admissible points that can model the \gls{IKD} associated with a given \gls{DMU}. By doing this for all \glspl{DMU}, we can achieve a large number of efficiency estimates per \gls{DMU}, allowing us to conduct statistical inference with them. The \gls{HR} procedure is commonly associated with the \gls{SMAA}, which, in turn, is associated with multiple criteria decision analysis rather than efficiency assessment; see \cite{YANG2012483}, for instance. \cite{LAHDELMA2006241}, and \cite{KADZINSKI20171} combined the \gls{SMAA} model with \gls{DEA}, exploring the space of multipliers optimized by the primal \gls{DEA}. This integration allows describing \glspl{DMU} in terms of rank acceptability indices, central weights,
and confidence factors, which can be useful because the \gls{DEA} technique does not discern efficient \glspl{DMU} among them. As quoted by \cite{LAHDELMA2006241}, ``\textit{these so-called non-parametric
methods }[\gls{DEA} and \gls{SMAA}] \textit{explore the weight space in order to identify weights favorable for each alternative} [\gls{DMU}].'' 


In this study:
\begin{enumerate}
    \item We propose an alternative that lies in defining a set (and its boundary) in the space containing all admissible points that can model the \gls{IKD} of a given \gls{DMU}. We can apply our model regardless of the epistemic or aleatory nature of the imperfect knowledge -- all rely on the appropriate definition of the set. We can achieve a large number of efficiency estimates per \gls{DMU}, after running our algorithm a sufficient number of times, and making a statistical inference with each iteration results. Differently from \cite{LAHDELMA2006241}, \cite{YANG2012483}, and \cite{KADZINSKI20171}, who applied \gls{SMAA} to the multipliers optimized by \gls{DEA}, we introduce stochastic nature directly on input-output data through the definition of sets representing the imperfect knowledge of such data.
    
    \item We propose to bound imperfectly known data into convex sets (although non-convexity is allowed for those sets, they do not seem natural choices and, besides, they require the adoption of rejection strategies of some generated points). Creating data bounds is usually preferable to the analyst/expert because they are intuitive in many applications. For instance, let us suppose that a given hospital report sets the number of treated inpatients as twelve thousand. It is hardly the truly measured number of inpatients seen in that hospital ward and, as it is well known, \gls{DEA} is quite sensitive to data quality \citep{CABRERA2018155}. It means that, because of rounding, the actual number of inpatients lies within 11,500 and 12,499. These values might, then, be set as bounds for the imperfectly known datum. Modeling \gls{IKD} through different set shapes, including hyper-boxes and hyper-ellipsoids, is also an exciting and useful exercise. We provide the parametric equations that rule the \gls{IKD} modeling, considering that the extremes of those shapes are the data sets built with the experts' help.
    
    \item We compare our proposed approach with some alternatives that, due to their simplicity, are traditionally employed to replace \gls{IKD} (including missing data). One is the interesting interval \gls{DEA} \citep{DESPOTIS200224,SMIRLIS20061}. It will be easy to conclude that it is a particular case of our approach. 
    
    \item We also explain how to perform statistical inference using the estimates of efficiency.
    
    \item Finally, we employ our proposed approach to study public hospitals' performance considering a set of inputs, and desirable and undesirable outputs. The case study was based on the Portuguese public hospitals using data from 2013 to 2016, considering production-related inputs-outputs and some outputs related to quality and access. Disregarding this kind of variable when assessing hospital performance is likely to result in a financial analysis not accounting for the social functions of the public hospitals \citep{FERREIRA2018}.
\end{enumerate}

The remainder of this paper is structured as follows. Section \ref{sec:Data Envelopment Analysis} presents some notation and briefly describes a \gls{DEA} model based on the distance from the observations to the empirical frontier. Section \ref{sec:DEA_Review} makes a review of the different ways used to deal with imperfect knowledge in \gls{DEA}. Section \ref{sec:HiT_Run} shows the basics of the \gls{HR} routine and some ways of modeling data imperfect knowledge. Section \ref{subsec: DEA under uncertainty: Applying the Hit Run routine} explains how the \gls{HR} routine can be integrated with \gls{DEA} for robust efficiency estimates assessment, compares the proposed method with other alternatives, explains how we can make statistical inference using the efficiency estimates, and makes some additional considerations regarding the integrated approach's robustness. Section \ref{sec:An empirical application} uses a dataset composed of 108 Portuguese public hospitals and an appropriate \gls{DEA} in the presence of undesirable outputs to test the proposed method. Finally, Section \ref{sec:Concluding remarks} concludes this paper.


\section{Data Envelopment Analysis}\label{sec:Data Envelopment Analysis}
\noindent This section introduces the basic notation, the \gls{DEA} model, and some of its variants.

\subsection{Basic notation and some definitions}\label{subsec:Notation}
\noindent This subsection comprises five paragraphs successively devoted to the problem definition, the raw data of a model, and the concept of efficiency followed by the efficient frontier. Finally, the cases of inefficient \glspl{DMU} and directional improvements are introduced. 

\vspace{.25cm}

\noindent \textit{Problem definition.} \gls{DEA} is a model to assess the \textit{technical efficiency} of observations called, in this context, \glspl{DMU}, through an \textit{input-output} transformation analysis \citep{Cooper07b}. \textit{Inputs} represent the resources consumed to produce some goods or to deliver some services, generically called \textit{outputs}. An \textit{observation} in the literature also designates a pair formed by both inputs and outputs. 

\vspace{.25cm}

\noindent \textit{Basic data.} Consider $\text{DMU}_1, \ldots, \text{DMU}_j,\ldots,\text{DMU}_n$ the set of \glspl{DMU} to be analyzed ($J$ will be used to denote the set of the \glspl{DMU}' indices). Let $x^j = (x_{1}^j,\ldots,x_{i}^j,$ $\ldots,x_{m}^j)^\top$ denote a vector, where the $m$ components are the inputs used by the unit $j$, for $j=1,\ldots,n$. Similarly, we define $X$ as the $m \times n$ input matrix, for all the \glspl{DMU} considered. We can also define $y^j = (y_{1}^j,\ldots,y_{r}^j,\ldots,y_{s}^j)^\top$ as the vector, where the $s$ components are the outputs produced by the unit $j$, for $j=1,\ldots,n$. $Y$ denotes the $s \times n$ output matrix, for all the \glspl{DMU} analyzed. 

\vspace{.25cm}

\noindent \textit{Efficient \glspl{DMU}.} In the classical model with constant returns to scale, the (technical) efficiency of a unit $k \in J$ is, roughly speaking, the relationship (ratio) between the weighted sum of the outputs and the weighted sum of the inputs. These (input and output) weights (also called the input and output \textit{multipliers}) are, in general, assessed through the resolution of a linear programming model. It allows to compare \gls{DMU} $k$ with all the \glspl{DMU} $j \in J$ (including \gls{DMU} $k$ itself). The process should be done for all \glspl{DMU} $j \in J$. \gls{DEA} optimizes those multipliers (subject to certain constraints) trying to maximize the efficiency of \gls{DMU} $k$ regarding the entire set $J$. More formally, there are two ways of defining efficiency. A \gls{DMU} $k$ is technically efficient concerning $J$ if for its levels of consumed inputs, $x^k$, no other \gls{DMU} produces more (desirable) outputs than $y^k$. Likewise, for its produced/delivered outputs, $y^k$, no other \gls{DMU} consumes fewer inputs than $x^k$.

\vspace{.25cm}

\noindent \textit{Efficient frontier.} For all technically efficient \glspl{DMU} $k\in J$, the vectors $y^k$ and $x^k$ as well as all the linear convex combinations of such vectors allow the construction of the efficient frontier $F$, where the coefficients, $\mu^k_1,\ldots,\mu^k_j,\ldots,\mu^k_n$, are non-negative and must fulfill the normalization condition $\sum_{j=1}^n \mu_j^k= 1$.

\vspace{.25cm}

\noindent \textit{Inefficient \glspl{DMU} and directional improvements.} Inefficient \glspl{DMU} do not belong to the frontier $F$. To improve their efficiency, they must be projected on $F$. Let $D^k$ denote the distance between the vectors $x^k$ and $y^k$, on the one hand, and the frontier $F$, on the other hand, following a path defined by two direction vectors, $\delta^x$ and $\delta^y$ (a directional vector $\delta$ can be formed from these two vectors). It is obvious that when $k$ belongs to $F$, it is technically efficient and $D^k = 0$. Otherwise, it is inefficient and $D^k > 0$. Classical models of \gls{DEA} usually assume either an input- or an output-orientation. There are three situations.
In the case of input-oriented models, $\delta = (x^k,0)$ , \textit{i.e.}, inputs might be reduced by a factor $(1-D^k) \in ~ ]0,1]$, keeping the same levels of outputs, to turn $k$ more efficient. 
In the case of output-oriented models, $\delta = (0,y^k)$ and outputs can be increased by a factor $1/(1+D^k) \geqslant 1$, inputs held. 
In a more generic case, we may define a vector $\delta$ with nonzero components. The resulting model is said to be \textit{directional}, as it allows the simultaneous inputs' contraction and outputs' expansion.

\vspace{.25cm}

Once projected on $F$, the DMU $k$ becomes featured by the so-called \textit{targets}, which are the corresponding optimal values for the inputs and outputs of $k$. In other words, targets characterize $F$  from the linear convex combinations of the observations associated with efficient \glspl{DMU}. These targets are denoted by $x^{\ast j}$ and $y^{\ast j}$, for $j\in J$. Targets depend essentially on the distance $D^j$ and on the improvement direction $\delta$. Efficient \glspl{DMU} $j\in J$ verify $x^{\ast j} = x^j$ and $y^{\ast j}=y^j$, because $D^j = 0$, regardless of $\delta$. However, if $k\in J$ is not efficient, $F$ outperforms $k$ and the following holds: $\exists_{i=1,\ldots,m}~ x^{\ast k}_i < x^k_i$ or $\exists_{r=1,\ldots,s}~ y^{\ast k}_r > y^k_r$. These inequalities can be transformed into equations using \textit{slacks}. There is a non-negative slack for each variable. Assuming that these slacks result from the product between the scalar $D^k$ (distance to the frontier) and the components of the directional vector, as well as a quantity that does not depend on these two factors, we have: $x^{\ast k} = x^k - D^k~\delta^x - \gamma^x$ and $y^{\ast k} = y^k + D^k~\delta^y + \gamma^y$.

	\subsection{A radial directional DEA model}\label{subsec:DEA_Model}
\noindent	The most popular version of a \gls{DEA} model assumes that the targets correspond to the weighted arithmetic mean of all observations: 
	 $
     x^{\ast k}_{i} = \sum_{j=1}^n \mu_j^k x^{j}_{i},~ i=1,\ldots,m, ~\text{and}~ y^{\ast k}_r = \sum_{j=1}^n \mu_j^k y^{j}_{r},~ r=1,\ldots,s,
     $
such that $\sum_{j=1}^n \mu_j^k = 1$. Therefore, optimizing targets means optimizing the weights $\mu^k$. Because of the inequalities ruling the relationship between targets and observations, we have:
	$
    \sum_{j=1}^n \mu_j^k x^{j}_{i} \leqslant x^k_i,~ i=1,\ldots,m, ~\text{and}~ \sum_{j=1}^n \mu_j^k y^{j}_{r} \geqslant y^k_r,~ r=1,\ldots,s.$ Assuming that both inputs and outputs are allowed to change radially to project $(x^k,y^k)$ on $F$, and also that such a change depends on a predefined path, $\delta=(\delta^x,\delta^y)$, the radial distance from $k$ to $F$ is a non-negative scalar ($D^k \in \mathbb{Z}$). It means that the previous inequalities can be rewritten as follows:    
    $
    \sum_{j=1}^n \mu_j^k x^{j}_{i} \leqslant x^k_i - D^k\delta^x_i,~ i=1,\ldots,m, ~\text{and}~ \sum_{j=1}^n \mu_j^k y^{j}_{r} \geqslant y^k_r +  D^k\delta^y_r,~ r=1,\ldots,s.$ Since we want to know what is the largest value of $D^k$ and possible slacks that keep feasible the previous systems of constraints, a linear problem can be stated as follows, \textit{vide} Equation (\ref{eq:directionalDEA}) \citep{Fukuyama17}.
 
    \begin{equation}
    \label{eq:directionalDEA}
    	\begin{array}{ll}
        	D^{\ast k} = \max  & D^k +  {\displaystyle  \varepsilon \left( \sum_{i=1}^m \gamma_i^x + \sum_{r=1}^s \gamma_r^y \right)}  \\
            \mbox{subject to:}    &  \\
       & {\displaystyle \sum_{j=1}^n \mu_j^k x^{j}_{i} + D^k\delta^x_i + \gamma_i^x = x^k_i, \;\, i = 1,\ldots,m,} \\
       &  {\displaystyle \sum_{j=1}^n \mu_j^k y^{j}_{r} - D^k\delta^y_r - \gamma_r^y = y^k_r, \;\, r = 1,\ldots,s,} \\
       &  {\displaystyle \sum_{j=1}^{n} \mu_j^k=1,} \\
       & {\displaystyle \mu_j^k \geqslant 0, \;\, j=1,\ldots, n,} \\
       & {\displaystyle \gamma_i^x \geqslant 0, \;\, i=1,\ldots, m,} \\
       & {\displaystyle \gamma_r^y \geqslant 0, \;\, r=1,\ldots, s.}
        \end{array}
    \end{equation}

\subsection{Data Envelopment Analysis with undesirable outputs}\label{subsec:Methodological issues}

\noindent Hitherto, the \gls{DEA} model \ref{eq:directionalDEA} has considered inputs and desirable outputs. However, some undesirable outputs are usually produced (in some cases, desirable outputs cannot be produced/delivered without the undesirable ones). Let $u^j = (u_{1}^j,\ldots,u_{h}^j,\ldots,u_{v}^j)^\top$ be a vector associated with the $p$ components representing the undesirable output levels produced by the \gls{DMU} $j$. As before, $U$ denotes the $p \times n$ undesirable output matrix, for all the \glspl{DMU} considered. In this case, targets associated with undesirable outputs are: $u_h^{\ast k} = \sum_{j=1}^n \mu_j^k u_h^j,~h=1,\ldots,v.$ It is straightforward to conclude that $u_h^{\ast k} \leqslant u_h^k$, for all $h\in\{1,\ldots,v\}$. Given the commonly assumed weak disposability over undesirable outputs, the previous relationship can be rewritten in terms of the distance $D^k$ and the components $\delta^u$ for the directional vector, $\delta=(\delta^x, \delta^y, \delta^u)$:

\begin{equation}\label{eq.u_model}
    \sum_{j=1}^n \mu_j^k u_h^j + D^k \delta_h^u = u_h^k,~h=1,\ldots,v.
\end{equation}

Equation (\ref{eq.u_model}) can be inserted into Model (\ref{eq:directionalDEA}) as a new constraint. Unfortunately, it does not correctly deal with undesirable outputs \citep{Kuosmanen2005, Kuosmanen2009}. The imposition of an equation related to the weak disposability in undesirable outputs is not sufficient. An abatement factor $\theta_j^k \in [0,1]$ should be applied to the intensities $\mu_j^k$ in the outputs-related constraints. It results into a nonlinear problem that demands for linearization. Let the intensities $\mu_j^k$ be partitioned into two non-negative factors, $\alpha_j^k$ and $\beta_j^k$, \textit{i.e.}, $\mu_j^k = \alpha_j^k + \beta_j^k$ for $\alpha_j^k,\beta_j^k \geqslant 0$ for all $j\in J$. If $\alpha_j^k = \theta_j^k \mu_j^k$ represents the part of \gls{DMU} $k$ remaining active, then $\beta_j^k = (1-\theta_j^k)\mu_j^k$ is the part of that \gls{DMU}'s output abated via scaling down of activity level \citep{Kuosmanen2005}. In light of this, we get the following model:

		\begin{equation}\label{eq:directionalDEA_corrected}
        \begin{array}{ll}
        	D^{\ast k} = \max & D^k + \varepsilon \left( \displaystyle \sum_{i=1}^m \gamma_i^x + \displaystyle \sum_{r=1}^s \gamma_r^y \right) \\
            \mbox{subject to:}    &  \\
       & {\displaystyle \sum_{j=1}^n \alpha_j^k ~x^{j}_{i} + \sum_{j=1}^n \beta_j^k  x^{j}_{i} + D^k \delta^x_i +  \gamma_i^x= x^k_i, \;\, i = 1,\ldots,m,} \\
       &  {\displaystyle \sum_{j=1}^n \alpha_j^k  y^{j}_{r} - D^k  \delta^y_r ~ - ~ \gamma_r^y= y^k_r, \;\, r = 1,\ldots,s,} \\
       & {\displaystyle \sum_{j=1}^n \alpha_j^k  u^{j}_{h} + D^k  \delta^u_h= u^k_h, \;\, h = 1,\ldots,p,} \\
       &  {\displaystyle \sum_{j=1}^{n} \alpha_j^k +  \sum_{j=1}^{n} \beta_j^k = 1,} \\
       & {\displaystyle \alpha_j^k, \beta_j^k \geqslant 0, \;\, j=1,\ldots, n,} \\
       & {\displaystyle \gamma_i^x \geqslant 0, \;\, i=1,\ldots, m,} \\
       & {\displaystyle \gamma_r^y \geqslant 0, \;\, r=1,\ldots, s.}
        \end{array}
\end{equation}

We use Model (\ref{eq:directionalDEA_corrected}) to estimate the Portuguese public hospitals' efficiency levels in our case study (\textit{vide infra}).


\section{Dealing with imperfect knowledge in DEA: A brief  review}\label{sec:DEA_Review}
\noindent \gls{DEA} and other models alike, require in general that all inputs and outputs are perfectly known. If this is not the case, then biased conclusions may arise because models are typically sensitive to data quality. The biasing degree may naturally depend on the extent of the \gls{IKD}. Several alternatives have been proposed in the literature, each with its advantages, shortcomings, and caveats.

It is essential to understand the problem of \gls{IKD} for efficiency assessment. There are several ways of modeling the \gls{IKD}. They can be classified into four distinct groups: (1) deletion of observations, (2) simple or pure substitution of observations, (3) more sophisticated substitutions, and (4) feasible sets of values.

\begin{enumerate}
\item \textit{Modeling through the deletion of \glspl{DMU} with imperfect data (Omission).} Perhaps, deletion of \glspl{DMU} is the most employed way of modeling \gls{IKD}, especially in exploratory analyses involving statistical tests or when the sample's size is substantial. However, we note that disregarding \glspl{DMU} from the efficiency analysis could be a pitfall, as they can be potential benchmarks. It decreases the statistical power of the conducted analyses. Despite these disadvantages, the deletion of \glspl{DMU} featured by \gls{IKD} is easy to implement and justified, making it so largely employed. Two main approaches to deal with omission are:

\begin{enumerate}[label=\itshape\alph*\upshape)]
    \item \textit{Listwise (complete case) deletion.} In this case, we check for IKD cases and remove the associated \glspl{DMU} from the analysis. The remaining sample is, then, used for the efficiency assessment through the traditional linear programming techniques. Listwise deletion may decrease the statistical power of the employed method, in any case.
  
    \item \textit{Pairwise deletion.} Unlike the previous case, pairwise deletion corresponds to the removal of a \gls{DMU} when its data is imperfectly known for a given variable and only if this one is under analysis. The same \gls{DMU} is, then, carried back into the analysis if that variable is no longer considered. This kind of approach makes \gls{DEA}-based models incomparable if they were based on the same dataset.
\end{enumerate}
  
  
\item \textit{Modeling through simple substitution of imperfectly known data (Imputation).} In this group, \gls{IKD} are replaced by \textit{appropriate} estimates. Since these estimates could not be the most appropriate ones, the procedure should be repeated a considerable number of times. Some approaches for dealing with imputation are:

\begin{enumerate}[label=\itshape\alph*\upshape)]
    \item \textit{Hot-deck imputation} \citep{Reilly1993}. This technique is based on the idea that similar \glspl{DMU} exhibit identical consumption and production profiles. Hence, we replace imperfectly known data with values copied from (randomly selected) similar observations.

    \item \textit{Cold-deck imputation.} As in the case of hot-deck imputation, in this case, one substitutes \gls{IKD} using similar observation but that belongs to another dataset.

    \item \textit{Mean imputation (or mean substitution)} \citep{Raaijmakers1999}. \gls{IKD} are replaced with the average of the considered variable. It does not change the mean of that variable. However, it carries out some problems due to the attenuation of correlations involving the imputed variable(s), thus being problematic in multivariate analyses.

    \item \textit{Regression} \citep{OLINSKY200353}. We can use a simple or multiple regression model with non-imperfect data to replace values featured by imperfect knowledge. Random noise can also be added to the estimate. Each fitted value is associated with a confidence interval translating the error of the estimate. The estimation is as good as any other value within that interval to substitute imperfect data. Moreover, the \textit{optimal/ideal} model is often challenging to achieve.

    \item \textit{Multiple imputation.} Each entry of \gls{IKD} is substituted many times using appropriate distributions, and generating as much different outcomes and analyses.
\end{enumerate}


\item \textit{Modeling through more sophisticated substitution (Analysis).} This kind of method mainly uses the maximum likelihood estimation to assess some relevant parameters to the analysis.

\begin{enumerate}[label=\itshape\alph*\upshape)]
    \item \textit{Expectation-maximization (Dempster-Laird-Rubin) algorithm.} This algorithm starts by determining the model parameters and by estimating potential alternative values for \gls{IKD} given current observations. Parameters are, then, refined admitting that \gls{IKD} are perfectly known. Using these parameters, one re-estimates the substitutes for \gls{IKD}. The process repeats until convergence is achieved.

    \item \textit{Maximum likelihood estimation.} Differently from the previous method, in this case one only determines the model parameters once and, if necessary, estimate then the acceptable substitutes for the cases of \gls{IKD}. 

\end{enumerate}

\item \textit{Modeling through feasible sets of values.} When imperfect knowledge cannot be modeled by using a single point (as it is quite restrictive and, unfortunately, frequent), a pleasing way is to consider sets of feasible points. For continuous intervals, there are two main approaches:

\begin{enumerate}[label=\itshape\alph*\upshape)]
    \item \textit{Interval \gls{DEA}} \citep{DESPOTIS200224,SMIRLIS20061}. This procedure establishes boundaries for data imperfect knowledge: $x^j \in [\underline{x}^j,\overline{x}^j]$ and $y^j \in [\underline{y}^j,\overline{y}^j]$ for $j\in J$. Of course, if data are perfectly known, then $\exists j\in J$ such that $\underline{x}^j=\overline{x}^j$ and $\underline{y}^j=\overline{y}^j$. Notwithstanding, one defines two scenarios: (W) worst, described by $\mathscr{S}^W=\{(x^j,y^j)\in \mathbb{R}_+^m\times\mathbb{R}_+^s ~|~ x^j = \overline{x}^j,~ y^j=\underline{y}^j,~j\in J\}$, and (B) best, with $\mathscr{S}^B=\{(x^j,y^j)\in \mathbb{R}_+^m\times\mathbb{R}_+^s ~|~ x^j = \underline{x}^j,~ y^j=\overline{y}^j,~j\in J\}$. One, then, may apply Model (\ref{eq:directionalDEA}) to project the worst version of \gls{DMU} $k$, $(\overline{x}^k,\underline{y}^k)$, into the frontier constructed using $\mathscr{S}^B$, to achieve the maximum distance of $k$ to the frontier under imperfect knowledge: $\overline{D}^{\ast k}$. Likewise, one may project the best version of $k$, $(\underline{x}^k,\overline{y}^k)$, into the frontier constructed with $\mathscr{S}^W$, and get the smallest distance of $k$ to the frontier: $\underline{D}^{\ast k}$. Therefore, one concludes that there is an interval associated with the efficiency of \gls{DMU} $k$: $D^{\ast k} \in [\underline{D}^{\ast k}, \overline{D}^{\ast k}]$. Although useful to fix boundaries for the efficiency of \glspl{DMU}, this alternative disregards the modeling of IKD, making both the analysis of efficiency distributions and the application of statistical tests impossible. 
    
    \item \textit{Fuzzy set \gls{DEA}} \citep{SOLEIMANIDAMANEH20061199, Kao2007, WU2009227,Emrouznejad2014, Lio2018}. One applies the fuzzy set theory to \gls{DEA}. In general, it is not possible to solve fuzzy \gls{DEA} models using linear programming solvers because the coefficients of such models are fuzzy sets. Estimating efficiency scores through this approach is usually difficult because "\textit{a large number of input variables in fuzzy logic could result in a significant number of rules that are needed to specify a dynamic model}" \citep{SHOKOUHI2010387}. \cite{ENTANI200232} used the interval DEA together with the fuzzy approach to rank \glspl{DMU}.
\end{enumerate}
\end{enumerate}

Omission and imputation are the most relevant and frequent ways of handling \gls{IKD} regarding the \gls{DEA} utilization for efficiency assessment. Whereas the omission of \glspl{DMU} from the dataset is an easy exercise, there are several alternatives to attribute (estimate) values when data is imperfectly known. The simple act of deleting \glspl{DMU} from the analysis is often sufficient to bias the results. As it is widely known, \gls{DEA} is prone to the so-called \textit{curse of dimensionality}. Thus, imputation can be seen as a better approach for efficiency assessment in the presence of \gls{IKD}. The following are some of the most relevant alternatives to deal with this problem in \gls{DEA} (and models alike); see \cite{Wen2015} for more details:

\begin{enumerate}[label=\itshape\alph*\upshape)]
    \item \textit{Blank entries} \citep{Kuosmanen2009_blank}. When one input or one output is missing from the dataset for a given \gls{DMU}, it can simply be replaced by a large value (\textit{big M}) or by zero, respectively, to mitigate the influence of \glspl{DMU} with missing data on the efficiency assessment of other observations. Although formulated regarding the case of blank entries, this approach could easily be extended to the general case of \gls{IKD}. However, we note that, despite its simplicity, this alternative is problematic when the number of imperfectly known cases is large, translating into biased results (substantial inefficiency levels). The problem exacerbates when there is a certain degree of knowledge even for the cases of \gls{IKD}. Moreover, efficiency estimates are not comparable among different \glspl{DMU} because they are no longer evaluated using the same basis (the same variables).

    \item \textit{\gls{DEA} with Halo effect} \citep{ZHA20136135}. In this approach, one uses the mean imputation for the cases of \gls{IKD}. Then, one estimates the efficiency scores and rank \glspl{DMU}. Finally, considering \gls{DMU} $k$ (with IKD), which is in position $\mathcal{P}^k$, one defines the interval of admissible values for imperfectly known data using the values of \glspl{DMU} in positions $\mathcal{P}^k\pm 1$. This approach resembles the hot-deck imputation. First, using the mean values to replace \gls{IKD} should produce inaccurate efficiency levels and, accordingly, biased ranks. Sustaining a whole procedure on potentially biased ranks does not seem correct. Second, the values observed for the \gls{DMU} positioned in rank $\mathcal{P}^k + 1$ (or $\mathcal{P}^k - 1$) do not necessarily fit the \gls{IKD} of \gls{DMU} $k$.

    \item \textit{Uncertain \gls{DEA}} \citep{EHRGOTT2018231}. This model determines the amount of uncertainty necessary to raise the efficiency score of a \gls{DMU} featured by \gls{IKD}. Unfortunately, it usually results on nonlinear models that are difficult to solve.

    \item \textit{Imprecise \gls{DEA}} \citep{Cooper1999, PARK2010289}, which also returns nonlinear models because data are imprecise. Some linearizations have been proposed in the literature, as in \cite{ZHU2003513}, after scale transformation and variable alternations or procedures that turn imprecise into exact data. The interval DEA of \cite{DESPOTIS200224} is an extension of the imprecise DEA.

    \item \textit{Robust optimization and \gls{DEA}} \citep{SHOKOUHI2010387, SALAHI201667}, which is based on the concept of uncertainty sets, a robust counterpart optimization, and the imposition of a probability bound for constraints violation. 

    \item \textit{Stochastic \gls{DEA}} \citep{SENGUPTA1992259}. It specifies a probability density function to model errors in data. According to \cite{OLESEN20162}, there are two main directions upon which the stochastic DEA has been developed:
    \begin{itemize}
        \item[$-$] one based on statistical (but restrictive) axioms defining a statistical model and a sampling process into the DEA framework that provides biased estimators of the actual frontier \citep{Banker93}; and
        
        \item[$-$] another, based on the theory of chance constraints, replaces data with DMU-specific distributions \citep{OLESEN2006}. \cite{OLESEN20162} pointed out that one may criticize this approach because no formal statistical model with a sampling process is specified, making it challenging to identify what it is being estimated.
    \end{itemize}
\end{enumerate}

In the next section, we propose an alternative based on the so-called \gls{HR} to replace \gls{IKD} using few simple mathematical operations and a linear program within a for/while cycle (during some iterations). This Monte-Carlo-like operation makes the proposed alternative very easy to implement and run. It does not suffer from the problems verified for blank entries, \gls{DEA} with Halo effect, fuzzy set \gls{DEA}, and uncertain \gls{DEA}. This alternative was inspired by the interval \gls{DEA}, the robust optimization with \gls{DEA}, and the stochastic \gls{DEA}. We establish boundaries for \gls{IKD} and draw observations within that boundary. However, our approach does not insert in any of the two directions of the stochastic DEA as identified before. Also, we do not specify any distribution function to model errors in data, which may not be realistic. Usually, there is no evidence to choose one type of distribution function. Unlike interval \gls{DEA}, our proposal allows to specify IKD modeling using appropriate sets and obtain a considerable number of efficiency estimates. These, in turn, are useful for the analysis of efficiency distributions and statistical inference. Notice that all potential values belonging to those sets are admissible observations for the cases of \gls{IKD}. Likewise, any points in the Euclidean space $\mathbb{R}_+^{m+s+v}$ outside the sets are not admissible observations. Therefore, our proposal is out of the scope of the fuzzy set \gls{DEA}.


\section{The Hit \& Run algorithm}\label{sec:HiT_Run}
\noindent In this section, we propose a \gls{HR} routine \citep{Smith1984,Belisle1993,Kaufman1998} to simulate values belonging to a bounded set. This routine is a straightforward and useful procedure to simulate feasible points within a bounded (either convex or not) set, $\Lambda$.

\subsection{A first illustrative example}
\label{subsec: Basics regarding the HitRun routine}	
\noindent Let us start with an example of a convex set in the Euclidean space $\mathbb{R}^2$. We show the \gls{HR} routine, step-by-step, based on the example of Figure \ref{fig:HitandRun}. 

\begin{enumerate}
    \item Define the constraints associated with the set $\Lambda$. In this example, we consider $\Lambda$ defined through the following five constraints: $(1)~x_1 + x_2 \leqslant 10$, $(2)~-5x_1 + x_2 \geqslant -10$, $(3)~-2x_1 +x_2 \leqslant 5$, $(4)~x_1\geqslant 0$, and $(5)~x_2 \geqslant 0$.
    
    \item Select an initial point within $\Lambda\in \mathbb{R}^z$, $P^{(0)}=(P^{(0)}_1,\ldots,P^{(0)}_l,\ldots,P^{(0)}_z)$. In this case, $\Lambda\in \mathbb{R}^2$ and $P^{(0)}=(2,6)$. Figure \ref{fig:1b} exhibits the bounded set and the starting point for this example.
    
    \item Randomly select a directional vector, $d^{(1)}$, with unitary Euclidean norm, say $d^{(1)} =\frac{(0.20,0.40)}{\sqrt{0.20^2 + 0.40^2}}$.
    
    \item Determine the distance, $\lambda^{(1)}$, between $P^{(0)}$ and the boundary of $\Lambda$. 
    \begin{enumerate}
        \item Define the point in the boundary as $\hat{P}^{(0)} = P^{(0)} + \lambda^{(1)} d^{(1)}$.
        
        \item Compute $\lambda^{(1)}_1$ associated with the first constraint defining $\Lambda$. The point $\hat{P}^{(0)}$ verifies $\hat{P}^{(0)}_1+\hat{P}^{(0)}_2=10$ or, equivalently, $P^{(0)}_1 + \lambda^{(1)}_1 d_1^{(1)} + P^{(0)}_2 + \lambda^{(1)}_1 d_2^{(1)} = 10 \Leftrightarrow \lambda^{(1)}_1 = \frac{10 - (P^{(0)}_1+P^{(0)}_2)}{d_1^{(1)}+d_2^{(1)}}$. In this case, $\lambda^{(1)}_1 = \frac{10 - (2 + 6)}{0.20 + 0.40} \sqrt{0.20^2 + 0.40^2} = \frac{10}{3} \sqrt{0.20^2 + 0.40^2}$.
        
        \item Compute $\lambda^{(1)}_2,\ldots,\lambda^{(1)}_5$ associated with the other four constraints defining $\Lambda$, similarly. See Figure \ref{fig:1a}. In this example, $\lambda^{(1)}_2 = 10\sqrt{0.20^2 + 0.40^2},~\lambda^{(1)}_3 = +\infty$ (because $(-2,1)^\top d^{(1)} = 0$), $\lambda^{(1)}_4=-10\sqrt{0.20^2 + 0.40^2},\lambda^{(1)}_5=-15\sqrt{0.20^2 + 0.40^2}$.
        
        \item Compute $\lambda^{(1)}=\displaystyle{\min_{\lambda_q^{(1)}\geqslant 0,~q=1,\ldots,5} \{\lambda^{(1)}_1,\lambda^{(1)}_2,\lambda^{(1)}_3,\lambda^{(1)}_4,\lambda^{(1)}_5 \}} = \frac{10}{3}\sqrt{0.20^2 + 0.40^2}$.
    \end{enumerate}  
    
    \item Project $P^{(0)}$ on the boundary of $\Lambda$ following the direction $d^{(1)}$, and obtain the point $\hat{P}^{(0)}=(\hat{P}^{(0)}_1,\hat{P}^{(0)}_2)$ in the boundary. In this example, $\hat{P}^{(0)} = (2,6)+\frac{10}{3}\sqrt{0.20^2 + 0.40^2}(0.20,0.40) = \frac{\sqrt{0.20^2 + 0.40^2}}{3}(8, 22)$.
    
    \item Define the line linking $P^{(0)}$ to $\hat{P}^{(0)}$, as $L^{(1)} = \Lambda \cap \{ P\in\mathbb{R}^z,~e \in [0,1] ~|~ P = P^{(0)} + \lambda^{(1)} d^{(1)} e\}$.
    
    \item Define a new point $P^{(1)}$ belonging to $L^{(1)}$, as follows:
    \begin{enumerate}
        \item Randomly (with replacement) generate a number $\xi\sim \text{uniform}(0,1)$. In this example, $\xi=0.60$.
        
        \item Compute the point $P^{(1)} = P^{(0)} + \lambda^{(1)} d^{(1)} e$, replacing $e$ by $\xi$.
    \end{enumerate}
    
    \item Repeat Steps 3-7 considering the previous point $P^{(1)}$ as the starting point. For example, the second iteration starts from this intermediary point, whose projection is $(\hat{P}_1^{(1)},\hat{P}_2^{(1)})=(16/5,6)$ because $\lambda^{(2)}=1.60$ and according to $d^{(2)}=\frac{(0.5,-0.5)}{\sqrt{0.5^2 + 0.5^2}}$. Choosing $\xi=0.75$, we get $(P_1^{(2)},P_2^{(2)})=(3,31/5)$. The result of a few more iterations is also shown in Figure \ref{fig:1b} using red dots.
\end{enumerate}

\begin{figure}[t]
		\centering
        \subfloat[Application]{
		\begin{tikzpicture}
		\begin{axis}[
		height=9cm,
		width=10cm,
		xmin=0,
		xmax=4,
		ymin=0,
		ymax=9,
		xlabel=$x_1$,
		ylabel=$x_2$,
		ytick={0,1,...,9},
		xtick={0,1,...,4},
		grid=major,
		]

		\addplot[thick, blue, name path=A] {5*x - 10};
		\addplot[thick, red, name path=B] {2*x + 5};
		\addplot[thick, black, name path=C] {-x + 10};
		\addplot+[draw=none, mark=none, name path=D] {0};
		\draw[dashed, thick] (axis cs:2,6) -- (axis cs:2.6667,7.3333);
		\node[label={135:{$P^{(0)}$}},circle,fill,inner sep=2pt, color=blue, mark=*] at (axis cs:2,6) {};
		\node[label={90:{$\hat{P}^{(0)}$}},circle,fill,inner sep=2pt, color=blue, mark=*] at (axis cs:2.6667,7.3333) {};
		\draw[dashed, thick] (axis cs:2,6) -- (axis cs:2.6667,7.3333);
		\draw[->, color=blue, thick] (axis cs:2,6) -- (axis cs:2.2,6.4);
		\draw[dashed, thick] (axis cs:2.4,6.8) -- (axis cs:3.2,6);
		\node[label={160:{$P^{(1)}$}},circle,fill,inner sep=2pt, color=blue, mark=*] at (axis cs:2.4,6.8) {};
		\node[label={0:{$\hat{P}^{(1)}$}},circle,fill,inner sep=2pt, color=blue, mark=*] at (axis cs:3.2,6) {};
		\draw[->, color=blue, thick] (axis cs:2.4,6.8) -- (axis cs:2.65,6.55);
		\node[label={190:{$P^{(2)}$}},circle,fill,inner sep=2pt, color=blue, mark=*] at (axis cs:3,6.2) {};
		
		\pgfplotstableread{inputdata.txt}\mydata;
		\addplot [color=red, circle,fill,inner sep=1pt, only marks,] table [x expr = \thisrowno{0}, y expr = \thisrowno{1}] {\mydata};
		
		\end{axis}
		\end{tikzpicture}\label{fig:1b}}
		\subfloat[Equations ruling the Hit \& Run framework for the present example, considering the generic iteration $\ell=1,\ldots,t$. ]{\label{fig:1a}\includegraphics[width = .3\columnwidth]{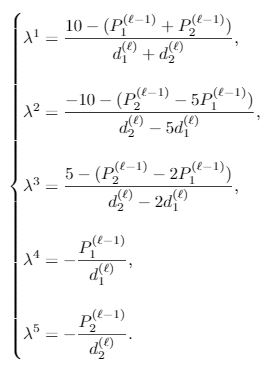}}
		\caption{Hit \& Run framework applied to a simple convex set.\label{fig:HitandRun}}
	\end{figure}
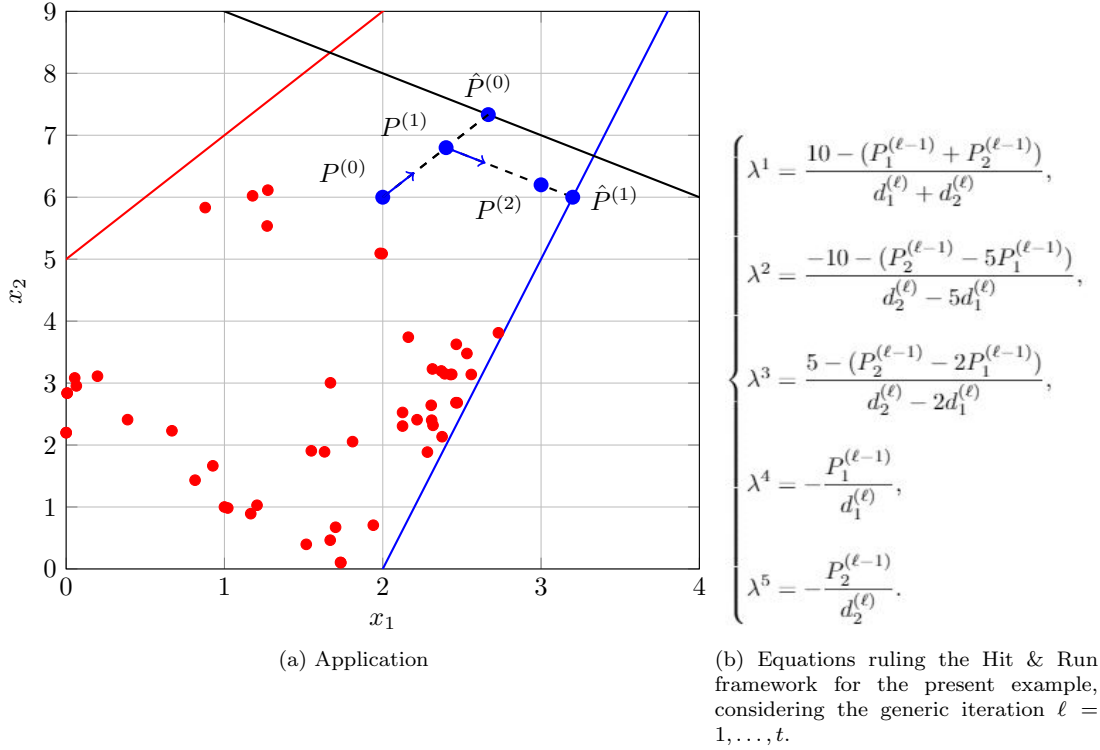

\subsection{The Hit \& Run algorithm for imperfect knowledge of data modeling}
\noindent Algorithm 1 in Appendix A synthesizes the HR algorithm for data generation (simulation) within the set $\Lambda$.\footnote{Appendix available online at: \url{https://drive.google.com/drive/folders/1jAmKFzz_PWyPKSNTxqO_0mM8BWKWn3-D?usp=sharing}}

\subsection{General computation of distance parameters for convex sets defined by linear constraints}
\noindent From the example in Subsection \ref{subsec: Basics regarding the HitRun routine} and Algorithm 1 (Appendix A), the importance of computing the parameter $\lambda$ related to the distance of a point to the frontier makes it clear. Of course, such a distance depends on the shape of $\Lambda$. This subsection describes the general computation of $\lambda$ parameters considering convex polytopes defined by several linear constraints. From now on, we consider that these sets must be convex \citep{GREGORY2011417}, in line with some robust optimization developments.

Let us consider a polytope defined by the intersection of $p$ linear constraints: 

\begin{equation}
    \sum_{l=1}^z a_{ql} x_l \leqslant b_q,\;\; q=1,\ldots,p, 
\end{equation}

\noindent such that the parameters $a_{ql}$ and $b_q$ are real numbers. One of the first iterations of \gls{HR} consists of defining a starting point $P^{(0)}$ and to project it on the boundary of $\Lambda$ following a directional vector $d^{(1)}$. Note that, in the boundary of $\Lambda$, the constraints of the polytope write as $\sum_{l=1}^z a_{ql} x_l = b_q,\;\; q=1,\ldots,p$. That being said, it is straightforward to conclude that $\sum_{l=1}^z a_{ql} (P^{(0)}_l + \lambda^{(1)}_q d^{(1)}_l) = b_q,\;\; q=1,\ldots,p$, being $\lambda^{(1)}_q$ the distance between $P^{(0)}$ and the boundary of $\Lambda$. This Equation can be rewritten as follows: 

\begin{equation}
    \lambda^{(1)}_q = \frac{b_q - \sum_{l=1}^z a_{ql}P^{(0)}_l}{\sum_{l=1}^z a_{ql}d^{(1)}_l}, \;\; q=1,\ldots,p.
\end{equation}

\noindent Note that there are $p$ values for these $\lambda$ parameters, and we select the smallest non-negative value among them:

\begin{equation}
    \lambda^{(1)} = \min_{\lambda\geqslant 0} \{ \lambda^{(1)}_1, \ldots, \lambda^{(1)}_q, \ldots, \lambda^{(1)}_p \}
\end{equation}

The \gls{HR} routine runs for a significant number of times, $t$. Hence, there are $p\cdot t$ distinct estimates for $\lambda$, denoted from now on by $\lambda^{(\ell)}_q$ for $q=1,\ldots,p$ and $\ell=1,\ldots,t$. The generic Equation ruling the estimation of these parameters is as follows:

\begin{equation}\label{eq:lambda_ell}
    \lambda^{(\ell)}_q = \frac{b_q - \sum_{l=1}^z a_{ql}P^{(\ell-1)}_l}{\sum_{l=1}^z a_{ql}d^{(\ell)}_l}, \;\; q=1,\ldots,p, \;\; \ell=1,\ldots,t.
\end{equation}

The parameter $\lambda^{(\ell)}$ associated with the $\ell$-th iteration of the HR algorithm is:

\begin{equation}\label{eq:lambda_min}
    \lambda^{(\ell)} = \min_{\lambda\geqslant 0} \{ \lambda^{(\ell)}_1, \ldots, \lambda^{(\ell)}_q, \ldots, \lambda^{(\ell)}_p \}, \;\; \ell=1,\ldots,t.
\end{equation}

The estimation of the parameter $\lambda^{(\ell)}$ naturally depends upon the shape of $\Lambda$. The next subsection presents some particular ways of modeling the \gls{IKD} through the definition of $\Lambda$.

\subsection{Particular cases of modeling data imperfect knowledge}
\label{subsec: Different ways of modeling uncertainty}
\noindent There are endless ways of modeling the sets describing the \gls{IKD}. The most natural one is, perhaps, assuming that it is well modeled by a box in the Euclidean space $\mathbb{R}^z$ \citep{Soyster1973}. However, other alternatives are available, including hyper-ellipsoids \citep{Ben98,Ben00}, hyper-rhombi or polyhedral \citep{Bertsimas2004}, or, more generic, super-ellipses \citep{Gielis2003}. The estimation of the parameter $\lambda^{(\ell)}$ associated with the $\ell$th iteration of the \gls{HR} procedure largely depends on the definition of $\Lambda$. The definition of this set's shape is not always straightforward, although some insights have been provided in the literature. For instance, \cite{Bertsimas2018} propose to use historical data and statistical estimates to define data-driven sets, while \cite{Bertsimas2009} rely on decision-maker risk preferences based on the theory of coherent risk measures of \cite{Artzner1999}.

\subsubsection{The case of hyper-box}
\noindent As said, the most simple way of modeling IKD is through a hyper-box in $\mathbb{R}^z$; see Definition \ref{proposition1}. Figure B.1a in Appendix B exhibits 10,000 points generated by HR for a two-dimension box with $x_1\in [2,8]$ and $x_2\in [3,7]$.

\begin{definition}[Hyper-box in $\mathbb{R}^z$]\label{proposition1}\normalfont The hyper-box whose set $\Lambda$ is defined by
\begin{equation}
\label{eq:Set_box}
\Lambda = \left\{ (x_1,\ldots,x_l,\ldots,x_z)\in\mathbb{R}^z~|~ 
x_1\in[\underline{x}_1,\overline{x}_1],\ldots,
x_l\in[\underline{x}_l,\overline{x}_l],\ldots,
x_z\in[\underline{x}_z,\overline{x}_z]      \right\}
\end{equation}
is associated with
\begin{equation}
\label{eq:Lambdas_box}
\lambda^{(\ell)} = \min \left\{ \min_{\lambda\geqslant 0} \left\{ \frac{\underline{x}_l - P_l^{(\ell -1)}}{d_l^{(\ell)}} ,~l=1,\ldots,z\right\}, \min_{\lambda\geqslant 0} \left\{\frac{\overline{x}_l - P_l^{(\ell -1)}}{d_l^{(\ell)}} ,~l=1,\ldots,z\right\}      \right\},~\ell=1,\ldots,t,
\end{equation}
which immediately holds from Equation (\ref{eq:lambda_ell}).
\end{definition}

\subsubsection{The case of hyper-ellipsoid}
\noindent Another simple way of modeling IKD is using a hyper-ellipsoid in $\mathbb{R}^z$, whose set $\Lambda$ obeys to:
\begin{equation}
\label{eq:Set_ellipsoid}
\Lambda = \left\{ (x_1,\ldots,x_z)\in\mathbb{R}^z~|~ 
\left(\frac{x_1 - P^{(0)}_1}{w_1}\right)^2 + \cdots +
\left(\frac{x_z - P^{(0)}_z}{w_z}\right)^2 = \sum_{l=1}^z \left(\frac{x_l - P^{(0)}_l}{w_l} \right)^2 \leqslant 1
\right\},
\end{equation}
where $w_1,\ldots,w_z$ are the semi-diameters or semi-axes of the hyper-ellipsoid contained in a box characterized by $x_1\in [P^{(0)}_1 - w_1, P^{(0)}_1 + w_1],\ldots,~ x_z\in [P^{(0)}_z - w_z, P^{(0)}_z + w_z]$. Through the hyper-ellipsoids centered in the (starting) point $P^{(0)}=(P^{(0)}_1,\ldots,P^{(0)}_l,\ldots,P^{(0)}_z)$ and with semi-axes $w_l~(\geqslant 0),~l=1,\ldots,z,$ we have the following Equation in the boundary:

\begin{equation}\label{eq:cond_ellip}
\sum_{l=1}^z \left(\frac{x_l - P^{(0)}_l}{w_l} \right)^2 = 1.
\end{equation}
We can define the hitting point $\hat{P}^{(0)}=(\hat{P}^{(0)}_1,\ldots,\hat{P}^{(0)}_l,\ldots,\hat{P}^{(0)}_z)$ in the boundary of $\Lambda$ through the centroid of $\Lambda$, a vector $d^{(1)}$ with unitary Euclidean norm, and a parameter $\lambda$: $\hat{P}^{(0)}_l = P^{(0)}_l + \lambda^{(1)} d^{(1)}_l,~l=1,\ldots,z$. Given Equation (\ref{eq:cond_ellip}), we have:

\begin{equation}
    \sum_{l=1}^z \left(\frac{(x_l - P^{(0)}_l)  +  d^{(1)}_l \lambda^{(1)}}{w_l} \right)^2 = 
\sum_{l=1}^z \frac{\left(x_l - P^{(0)}_l\right)^2  + \left(d^{(1)}_l\right)^2 \left(\lambda^{(1)}\right)^2  + 2 (x_l - P^{(0)}_l) d^{(1)}_l \lambda^{(1)}}{(w_l)^2}  = 1, 
\end{equation}

\noindent which can be rearranged as:

\begin{equation}\label{eq:split_ellip}
\begin{split}
&\underbrace{\sum_{l=1}^z \left( \frac{d^{(1)}_l}{w_l} \right)^2}_{=\varphi^{(1)}} \left(\lambda^{(1)}\right)^2 +  \underbrace{\left( 2 \sum_{l=1}^z \frac{(x_l - P^{(0)}_l) d^{(1)}_l}{(w_l)^2}  \right)}_{=\psi^{(1)}} \lambda^{(1)} + \underbrace{\sum_{l=1}^z \left( \frac{x_l - P^{(0)}_l}{w_l} \right)^2 -1}_{=\varsigma^{(1)}} =\\ &=\varphi^{(1)}\left(\lambda^{(1)}\right)^2  + \psi^{(1)} \lambda^{(1)} + \varsigma^{(1)} = 0 \Longleftrightarrow\\
&\Longleftrightarrow \lambda^{(1)} = \left[-\psi^{(1)} \pm \sqrt{\left(\psi^{(1)} \right)^2 - 4 \varphi^{(1)}\varsigma^{(1)}}\right]/2\varphi^{(1)},
\end{split}
\end{equation}
which holds after using Bhaskara's (quadratic) formula. Equation (\ref{eq:cond_ellip}) is quadratic. Hence, at the most, there are two different solutions for $\lambda^{(\ell)},~\ell=1,\ldots,t,$ in Equation (\ref{eq:lambda_ellipsoid}). Therefore, according to Equation (\ref{eq:split_ellip}), we select the smallest strictly positive solution for $\lambda^{(\ell)}$. 

Figure B.1b portrays the result of 10,000 iterations of the \gls{HR} procedure applied to the two-dimensional ellipse, centered in the starting point $P^{(0)}=(5,5)$ with $w_1=3$ and $w_2=1$. If $\underline{x}_l = P_l^{(0)} - w_l$ and $\overline{x}_l = P_l^{(0)} + w_l$, then we can observe that the hyper-ellipsoid is more restrictive than hyper-boxes to model \gls{IKD}.

The parameters $\lambda^{(\ell)}$ for the hyper-ellipsoid are in Definition \ref{proposition2}. It is immediate to conclude that Equation (\ref{eq:lambda_ellipsoid}) results from the generalization of $\lambda^{(1)}$ in Equation (\ref{eq:split_ellip}) for the $\ell$-th iteration of the \gls{HR} procedure.

\begin{definition}
[Hyper-ellipsoid in $\mathbb{R}^z$]\label{proposition2}\normalfont 
The hyper-ellipsoid defined by Equation (\ref{eq:Set_ellipsoid}) is associated with
\begin{equation}\label{eq:lambda_ellipsoid}
\lambda^{(\ell)} = \frac{-\psi^{(\ell)} \pm \sqrt{\left(\psi^{(\ell)} \right)^2 - 4 \varphi^{(\ell)}\varsigma^{(\ell)} }}{2\varphi^{(\ell)}},~\text{s.t. }\begin{cases}
\varphi^{(\ell)} =  \displaystyle\sum_{l=1}^z \left( \frac{d^{(\ell)}_l}{w_l} \right)^2\\
\\
\psi^{(\ell)} = \displaystyle 2 \sum_{l=1}^z \frac{(P^{(\ell-1)}_l - P^{(0)}_l)~ d^{(\ell)}_l}{(w_l)^2}\\
\\
\varsigma^{(\ell)} = -1+\displaystyle\sum_{l=1}^z \left( \frac{P^{(\ell-1)}_l - P^{(0)}_l}{w_l} \right)^2
\end{cases}.
\end{equation}
\end{definition}

\subsubsection{The case of hyper-rhombus}
\noindent Finally, we can also consider the hyper-rhombi for IKD modeling, as shown in Figure B.1c for the rhombus in $\mathbb{R}^2$ centered in $P^{(0)}=(5,5)$. Generically, a hyper-rhombus centered in $P^{(0)}$ and has a set $\Lambda$ defined by:
\begin{equation}
\label{eq:hyper_rombuses_Lambda}
\Lambda = \left\{ (x_1,\ldots,x_z)\in\mathbb{R}^z~|~ 
\displaystyle\sum_{l=1}^z\left| \frac{x_l - P_l^{(0)}}{w_l} \right| \leqslant 1 \right\}.
\end{equation}

Let us consider a (two-dimension) rhombus centered in $P^{(0)}=(P^{(0)}_1,P^{(0)}_2)$ and with semi-axis $2w_1$ and $2w_2$. This rhombus is featured by $\left|\frac{x_1 - P_1^{(0)}}{w_1} \right| + \left|\frac{x_2 - P_2^{(0)}}{w_2} \right|\leqslant 1$, which translates into four constraints ($p=4$): $\frac{x_1- P_1^{(0)}}{w_1} + \frac{x_2- P_2^{(0)}}{w_2} \leqslant 1$, 
$\frac{x_1- P_1^{(0)}}{-w_1} + \frac{x_2- P_2^{(0)}}{w_2} \leqslant 1$, 
$\frac{x_1- P_1^{(0)}}{w_1} + \frac{x_2- P_2^{(0)}}{-w_2} \leqslant 1$, and 
$\frac{x_1- P_1^{(0)}}{-w_1} + \frac{x_2- P_2^{(0)}}{-w_2} \leqslant 1$. These constraints can be rewritten as follows:
$\frac{x_1}{w_1} + \frac{x_2}{w_2} \leqslant 1 + \frac{P_1^{(0)}}{w_1} + \frac{P_2^{(0)}}{w_2}$, 
$\frac{x_1}{-w_1} + \frac{x_2}{w_2} \leqslant 1+ \frac{P_1^{(0)}}{-w_1} + \frac{P_2^{(0)}}{w_2}$, 
$\frac{x_1}{w_1} + \frac{x_2}{-w_2} \leqslant 1+ \frac{P_1^{(0)}}{w_1} + \frac{P_2^{(0)}}{-w_2}$, and 
$\frac{x_1}{-w_1} + \frac{x_2}{-w_2} \leqslant 1+ \frac{P_1^{(0)}}{-w_1} + \frac{P_2^{(0)}}{-w_2}$, respectively. We have rewritten these constraints in the form of $\sum_{l=1}^2 a_{ql} x_l \leqslant b_q$, with $b_q = 1 + \sum_{l=1}^2 a_{ql} P_l^{(0)}$ for $q=1,\ldots,4$. From Equation (\ref{eq:lambda_ell}) we know that the parameter $\lambda$ associated with the $\ell$-th iteration is as follows:

\begin{equation}
    \lambda^{(\ell)}=\min_{\lambda\geqslant 0} \left\{  \lambda_q = \frac{1 + \displaystyle\sum_{l=1}^2 a_{ql} P_l^{(0)} - \displaystyle\sum_{l=1}^2 a_{ql} P_l^{(\ell-1)}}{\displaystyle\sum_{l=1}^2 a_{ql} d_l^{(\ell)}} = \frac{1 - \displaystyle\sum_{l=1}^2 a_{ql} \left(P_l^{(\ell-1)} - P_l^{(0)}  \right)}{\displaystyle\sum_{l=1}^2 a_{ql} d_l^{(\ell)}},~   q=1,\ldots,4 \right\},
\end{equation}

\noindent where $a_{ql} = 1/[\text{sgn}(a_{ql}) w_l]$ (the sign associated with $a_{ql}$ depends on the constraint. Still, it is clear that $|\alpha_{ql}|=1/w_l$ for $l=1,2$ and $q=1,\ldots,4$). In the present case ($\Lambda\in\mathbb{R}^2$), we have
\begin{enumerate}
\item[] $\lambda_1 = \left( 1 - \left( \frac{1}{w_1}(P_1^{(\ell-1)} - P_1^{(0)}) + \frac{1}{w_2}(P_2^{(\ell-1)} - P_2^{(0)})\right)\right)/\left( \frac{1}{w_1}d_1^{(\ell)} + \frac{1}{w_2}d_2^{(\ell)}  \right),$
\item[] $\lambda_2 = \left( 1 - \left( \frac{1}{-w_1}(P_1^{(\ell-1)} - P_1^{(0)}) + \frac{1}{w_2}(P_2^{(\ell-1)} - P_2^{(0)})\right)\right)/\left( \frac{1}{-w_1}d_1^{(\ell)} + \frac{1}{w_2}d_2^{(\ell)}  \right),$
\item[] $\lambda_3 = \left( 1 - \left( \frac{1}{w_1}(P_1^{(\ell-1)} - P_1^{(0)}) + \frac{1}{-w_2}(P_2^{(\ell-1)} - P_2^{(0)})\right)\right)/\left( \frac{1}{w_1}d_1^{(\ell)} + \frac{1}{-w_2}d_2^{(\ell)}  \right),$
\item[] $\lambda_4 = \left( 1 - \left( \frac{1}{-w_1}(P_1^{(\ell-1)} - P_1^{(0)}) + \frac{1}{-w_2}(P_2^{(\ell-1)} - P_2^{(0)})\right)\right)/\left( \frac{1}{-w_1}d_1^{(\ell)} + \frac{1}{-w_2}d_2^{(\ell)}  \right).$
\end{enumerate}
Determining the parameter $\lambda^{(\ell)}$ is identical if $z\geqslant 2$. There are $p=2^z$ constraints to describe the hyper-rhombus in $\mathbb{R}^z$. These constraints can be written as $\sum_{l=1}^z a_{ql} x_l \leqslant b_q$, with $b_q = 1 + \sum_{l=1}^z a_{ql} P_l^{(0)}$ and $a_{ql} = 1/[\pm w_l]$, for $q=1,\ldots,2^z$.

Definition \ref{proposition3} presents the Equation ruling the computation of the parameters $\lambda^{(\ell)}$ for the hyper-rhombi in $\mathbb{R}^z$. 

\begin{definition}[Hyper-rhombus in $\mathbb{R}^z$]\label{proposition3}\normalfont
The hyper-rhombus whose set $\Lambda$ is defined by Equation (\ref{eq:hyper_rombuses_Lambda}) verifies:
\begin{equation}
\lambda^{(\ell)}=\min_{\lambda\geqslant 0} \left\{  \lambda_q = \frac{1 - \displaystyle\sum_{l=1}^z a_{ql} \left(P_l^{(\ell-1)} - P_l^{(0)}  \right)}{\displaystyle\sum_{l=1}^z a_{ql} d_l^{(\ell)}},~   q=1,\ldots,2^z \right\},~\ell=1,\ldots,t,
\end{equation}
where $a_{ql} = 1/\pm w_l$, for $q=1,\ldots,2^z$.
\end{definition}


\section{DEA under the imperfect knowledge of data: Applying the Hit \& Run routine}
\label{subsec: DEA under uncertainty: Applying the Hit Run routine}	

\noindent This section explains how the \gls{HR} routine can be integrated with \gls{DEA} for robust efficiency estimates assessment, compares the proposed method with other alternatives, provides insights on how we can make statistical inference using the efficiency estimates, and makes some additional considerations regarding the integrated approach's robustness.

\subsection{The integrated algorithm HR+DEA}
\noindent Using the input and output matrices associated with the $n$ \glspl{DMU} under evaluation, a benchmarking exercise's main objective consists of estimating the $n$ efficiency scores (or, similarly, distances to the frontier). It is possible if inputs and outputs are known perfectly. Otherwise, in the presence of IKD, a solution can be used to estimate $t$ possible efficiency scores \textit{per} \gls{DMU}, \textit{i.e.}, by randomly generating a $n\times t$ matrix of possible scores, $E$, where $t$ is the number of realizations considered for the unknown values.

Indeed, due to the IKD, unobserved or incorrect input/output variables for a particular \gls{DMU} $k\in J$ can be replaced by a set $\Lambda^k$ in the $\mathbb{R}_+^m \times \mathbb{R}_+^s \times \mathbb{R}_+^v$-Euclidean space. As previously mentioned, this set can assume different shapes, such as hyper-boxes or hyper-ellipsoids. 

In a perfect knowledge situation, $\Lambda^k \longmapsto (x^k,y^k,u^k)$, \textit{i.e.}, the set reduces to a single point in the same Euclidean space. Otherwise, any point within $\Lambda^k$ is a possible candidate to be a potential observation for \gls{DMU} $k$. Since $\Lambda^k$ is a continuous space, there are infinite potential candidates for $k$, and the same does apply for each $j\in J$ concerning $\Lambda^j$. It means that achieving all potential candidates to cover the entire set requires a simulation routine with infinite iterations, $t$, like the one of \gls{HR}; see Algorithm 1 in Appendix A.\footnote{Appendix available at: \url{https://drive.google.com/drive/folders/1jAmKFzz_PWyPKSNTxqO_0mM8BWKWn3-D?usp=sharing}} Nonetheless, given the practical limitations of programming tools, the number of iterations must be finite, $t<+\infty$, although this parameter should be large enough to get satisfactory results. We expect no potential accuracy gains due to the increase of $t$ beyond a sufficiently large threshold (typically, a few thousand).

Before starting the \gls{HR} routine, it is necessary to define $n$ initial points, $(x^{j(0)}, y^{j(0)}, u^{j(0)})$ for each $j\in J$. These initial points must belong to the corresponding set $\Lambda^j$. We can, thus, pick any initial or starting point from $\Lambda^j$. For the sake of simplicity, the centroids are chosen for such a purpose. Using the centroids points, we can thus construct new inputs and outputs matrices, $(X^{(0)}, Y^{(0)}, U^{(0)})$, and estimate the 0-order frontier, $F^{(0)}$, as well as the corresponding efficiency scores, $D^{j(0)}$, for all $j\in J$, \textit{e.g.}, by using Model (\ref{eq:directionalDEA}) or Model (\ref{eq:directionalDEA_corrected}).

Consider a \gls{DMU} $k\in J$, whose efficiency score concerning the frontier $F$ we want to estimate. The \gls{HR} routine imposes that, for each iteration, $\ell$, $n$ new points $(x^{j(\ell)}, y^{j(\ell)}, u^{j(\ell)})\in \Lambda^j,~j\in J$, must be estimated. There are as many possible frontiers, $F^{(\ell)}$, as many iterations, $t$, we have defined at the beginning. Accordingly, there are, at most, $t$ possible efficiency estimates for \gls{DMU} $k$. To estimate those points, we propose the following procedure. For each iteration $\ell$: 

\begin{enumerate}
    \item Define a vector $d^{(\ell)}_j = (d^{x(\ell)}_j, d^{y(\ell)}_j, d^{u(\ell)}_j) $ for all $j\in J$ and such that $\| d^{(\ell)}_j \|_2 = 1$. This vector's entries are randomly selected (with replacement) from the uniform distribution bounded by --1 and +1.
    
    \item Following the path defined by $d^{(\ell)}_j$, estimate the distance $\lambda_j^{(\ell)}$ \textit{per} \gls{DMU} $j\in J$, according to the appropriate Equation provided in Subsection \ref{subsec: Different ways of modeling uncertainty}. This one is, thus, the distance between the point $(x^{j(\ell-1)}, y^{j(\ell-1)}, u^{j(\ell-1)})\in \Lambda^j,~j\in J$, that was achieved in the previous iteration, and the boundaries of $\Lambda^j$.
    
    \item Construct the line $L_j^{(\ell)}$ for each $j\in J$: 

        \begin{equation}
        L_j^{(\ell)} = \Lambda^j \cap \left\{ 
        \begin{pmatrix} 
        x^{j(\ell)} \\
        y^{j(\ell)} \\
        u^{j(\ell)}
        \end{pmatrix}^\top\in\mathbb{R}_+^{m+s+v},~ e \in [0,1]~|~ \begin{pmatrix} 
        x^{j(\ell)} \\
        y^{j(\ell)} \\
        u^{j(\ell)}
        \end{pmatrix}^\top = \begin{pmatrix} 
        x^{j(\ell-1)} \\
        y^{j(\ell-1)} \\
        u^{j(\ell-1)}
        \end{pmatrix}^\top + \lambda_j^{(\ell)} \begin{pmatrix} 
        d^{x(\ell)}_j \\
        d^{y(\ell)}_j \\
        d^{u(\ell)}_j
        \end{pmatrix}^\top e \right\}  ;
        \end{equation}

    \item Randomly and with reposition, generate $n$ quantities $\xi_j^{(\ell)}$ by following a predefined probability density function, typically uniform. The value of $\xi_j^{(\ell)}$ should be bounded by 0 and 1.
    
    \item Using the previous points $(x^{j(\ell-1)}, y^{j(\ell-1)}, u^{j(\ell-1)})$, the parameters $\xi_j^{(\ell)}\sim f_j$, the vectors $d^{(\ell)}_j$, and the distances $\lambda_j^{(\ell)}$, estimate the set of $n$ new points  $$( x^{j(\ell)}, y^{j(\ell)}, u^{j(\ell)}) =  ( x^{j(\ell-1)}, y^{j(\ell-1)}, u^{j(\ell-1)})  + \lambda_j^{(\ell)} d^{(\ell)}_j \xi_j^{(\ell)} ,~j\in J.$$
\end{enumerate}

Once the set of new points have been estimated, we can then construct the new matrices, $(X^{(\ell)}, Y^{(\ell)}, U^{(\ell)})$, as well as the new frontier, $F^{(\ell)}$, and estimate the distance of $(x^{k(\ell)}, y^{k(\ell)}, u^{k(\ell)})$ concerning $F^{(\ell)}$, \textit{i.e.}, $D^{k(\ell)}$. Figure B.3 in Appendix B presents an illustrative example of the \gls{HR}+\gls{DEA} integrated approach.\footnote{Appendix available at: \url{https://drive.google.com/drive/folders/1jAmKFzz_PWyPKSNTxqO_0mM8BWKWn3-D?usp=sharing}}

Note that there are $t$ possibly different estimates of $D^{k(\ell)}$ because $\xi_j^{(\ell)}$ depends on $j\in J$ and the probability density function, $f_j$, which can also be different for each \gls{DMU}. Algorithm 2 (in Appendix A) synthesizes the previous details on integrating the \gls{HR} routine with \gls{DEA}.

\subsection{Some interesting results about the integrated algorithm}
\noindent Given the formulation of the integrated approach, the following propositions hold. Proposition \ref{prop: simplification} states that, under the perfect knowledge case for all \glspl{DMU}, the integrated approach \gls{HR}+\gls{DEA} delivers the exact same efficiency scores (or distances) as \gls{DEA}. Therefore, \gls{DEA} applied to perfect knowledge of data is a particular case of this integrated approach.

\begin{proposition}\label{prop: simplification}\normalfont
The \gls{HR}+\gls{DEA} produces the same frontier (and, hence, the same efficiency estimates) as \gls{DEA} if $\Lambda^j \longmapsto (x^j,y^j)$ for all $j\in J$, regardless of the number of iterations, $t$.
\end{proposition}

\begin{proof}
\small It is straightforward to conclude that $(\forall j\in J,~\Lambda^j \longmapsto (x^j,y^j)) \Longrightarrow \lambda_j^{(\ell)}=0 \Longrightarrow ( x^{j(\ell)}, y^{j(\ell)}) =  ( x^{j(\ell-1)}, y^{j(\ell-1)})=\cdots=(x^j,y^j)\Longrightarrow F^{(\ell-1)}=F^{(\ell)}\Longrightarrow D^{j(\ell-1)}=D^{j(\ell)}=\cdots=D^{j(0)},~\forall \ell=1,\ldots,t$.
\end{proof}

Proposition \ref{prop: interval DEA} states that the interval \gls{DEA} approach, as proposed by \cite{DESPOTIS200224}, is a particular case of the proposed integrated approach \gls{HR}+\gls{DEA}. In that case, we have to impose $t\longmapsto+\infty$ to achieve the efficiency distributions' extremes. Still, in practice it is not possible to impose an infinite loop for the \gls{HR} procedure. Even if it would be possible, we often truncate those distributions; thus, the limits of interval \gls{DEA} will rarely be obtained in \gls{HR}+\gls{DEA}. It implies that the interval's width for efficiency estimates derived by \gls{HR}+\gls{DEA} is smaller than (or equal to) the width achieved by interval \gls{DEA}. 

\begin{proposition}\label{prop: interval DEA}\normalfont
If $t\longmapsto +\infty$ and $\Lambda^j$ is modeled using hyper-boxes for all $j\in J$, then the extremes of the distributions associated with $E^k,~k\in J,$ are equal to the ones derived using interval \gls{DEA}.
\end{proposition}

\begin{proof}
\small Let $\mathscr{E}^{k(t)}=\{ D^{j[1]},\ldots, D^{j[\ell]},\ldots, D^{j[t]}\}$ be the list of sorted efficiency scores (or distances) achieved using $t$ loops, such that $D^{j[1]}\leqslant\cdots\leqslant D^{j[\ell]}\leqslant\cdots\leqslant D^{j[t]}$. If $\Lambda^j$ is modeled using hyper-boxes for all $j\in J$, then $\max \mathscr{E}^{k(t)}~(=D^{k[t]})$ for $t\longmapsto+\infty$ corresponds to project $(\overline{x}^k, \underline{y}^k)$ in the frontier constructed using $\mathscr{S}^W=\{(x^j,y^j)\in \mathbb{R}_+^m\times\mathbb{R}_+^s ~|~ x^j = \overline{x}^j,~ y^j=\underline{y}^j,~j\in J\}$ (worst scenario of interval \gls{DEA}). Likewise, $\min \mathscr{E}^{k(t)}~(=D^{k[1]})$ is the distance of $(\underline{x}^k, \overline{y}^k)$ regarding $\mathscr{S}^B=\{(x^j,y^j)\in \mathbb{R}_+^m\times\mathbb{R}_+^s ~|~ x^j = \underline{x}^j,~ y^j=\overline{y}^j,~j\in J\}$ (best scenario of interval \gls{DEA}).
\end{proof}

\subsection{On comparing the integrated approach with other alternatives}\label{subsec:On comparing the Hit Run approach with other alternatives}
\noindent As previously mentioned, there are few alternatives to account for data imperfect knowledge, either in analyses based on \gls{DEA} or in any statistical analysis. Regarding \gls{DEA}, imputation remains one of the preferable alternatives given the dimensionality problem that affects non-parametric benchmarking models. To compare our alternative with others, we ran a Monte-Carlo simulation with 150 iterations. The chosen alternatives were: (\textit{i}) mean imputation, (\textit{ii}) hot-deck imputation, (\textit{iii}) interval \gls{DEA}, (\textit{iv}) regression using linear functions, and (\textit{v}) \gls{HR} with $t=100$ iterations using (\textit{v\textsubscript{1}}) hyper-boxes, (\textit{v\textsubscript{2}}) hyper-ellipsoids, and (\textit{v\textsubscript{3}}) hyper-rhombuses to model data imperfect knowledge.

For each Monte-Carlo iteration, we generated 300 \glspl{DMU} consuming two inputs, $x_1\sim 10+5\cdot\text{uniform}(0,1)$ and $x_2\sim 20+10\cdot\mathcal{N}^+ (0,1)$,\footnote{$\mathcal{N}(\rho,\sigma)$ stands for the Gaussian probability density function with mean $\rho$ and standard deviation $\sigma$. $\mathcal{N}^+ (\rho,\sigma)$ represents the strictly positive observations of the same density function.} and producing one output, $y_1$, which is a function of both $x_1$ and $x_2$. Three scenarios were constructed: 

(I) $y_1 \sim 5 \cdot (x_1)^{0.5}\cdot (x_2)^{0.7}\cdot \text{uniform}(0,1)$; 

(II) $y_1 \sim (2\cdot x_1 + 4\cdot x_2 + x_1\cdot  x_2)\cdot \mathcal{N}^+ (0,\frac{1}{3})$; and 

(III) $\log y_1 \sim \frac{1}{4} + \frac{1}{5}\cdot \log x_1 + \log x_2 + \frac{4}{10}\cdot \log^2 x_1 + \frac{1}{10}\cdot \log^2 x_2 + \frac{3}{10}\cdot \log x_1 \cdot \log x_2 + \log \text{uniform}(0,1)$. 

\noindent Then, we estimated the efficiency scores of the whole sample using the output-oriented \gls{DEA} (Model (\ref{eq:directionalDEA}) with $\delta^k=(0,y^k)$). 

Eighty gaps were randomly introduced into the dataset. The intervals used for the alternatives (\textit{iii}) and (\textit{v}) were randomly generated: $w_i^j = x_i^j + 10^{O_i -1}\rho_i,~i=1,2,$ and $w_r^j = y_r^j + 10^{O_r -1}\varpi_r,~r=1,$ where $O$ is the magnitude of the variable and $\rho_i,\varpi_r\sim\text{uniform}(0,1)$. 

We re-estimated the efficiency scores under the presence of \gls{IKD} through the eight imputation techniques above. To compare the former with the original efficiency scores, we used the following statistics: 

\begin{itemize}
    \item[$-$] Pearson's correlation coefficient (the larger, the better);
    \item[$-$] Kendall's correlation coefficient (the larger, the better);
    \item[$-$] Mean absolute error (MAE) (the smaller, the better);
    \item[$-$] Mean signed difference (MSD) (the smaller, the better).
\end{itemize}

Table \ref{tab:results_MonteCarlo} presents the main results obtained for the three scenarios and the eight alternatives, using two correlation coefficients and two error metrics. It is immediate to conclude that neither regression nor hot-deck imputations are good alternatives, as they minimize both correlation coefficients and maximize both error metrics (in absolute value). Using a mathematical model to regress data when the production function is unknown is likely to produce biased estimates to substitute imperfectly known data. A similar argument can be used for the imputation based on identical \glspl{DMU} (hot-deck). We also observe that the three \gls{HR} alternatives produce similar outcomes and are equally good at modeling \gls{IKD}. Furthermore, we verify that correlations (resp. error metrics) are larger (resp. smaller) in \gls{HR} compared to the mean imputation or the interval \gls{DEA}. However, it may not be sufficient to justify the adoption of an alternative like \gls{HR}, especially when the other two alternatives are simpler. We remark, though, that \gls{HR} can introduce a stochastic nature into \gls{DEA}, allowing statistical inference, as detailed in the next subsections. 

\begin{table}[t]
\begin{minipage}{\columnwidth}
\def\arraystretch{1.5}
  \centering
  \caption{On comparing the Hit \& Run approach with other alternatives.}
  \resizebox{\columnwidth}{!}{
    \begin{tabular}{ccrrrrccrrrr}
    \hline
    Correlation &  & Sc. (I) & Sc. (II) & Sc. (III) & Mean & Error &       & Sc. (I) & Sc. (II) & Sc. (III) & Mean \\
    \hline
    \multirow{7}[0]{*}{Pearson} & (\textit{i})   & 0.9564 & 0.9684 & 0.9793 & 0.9681 & \multirow{7}[0]{*}{MAE~\footnote{Mean absolute error}} & (\textit{i})   & 1.5197 & 5.2629 & 1.7957 & 2.8594 \\
          & (\textit{ii})  & 0.6620 & 0.6931 & 0.6611 & 0.6721 &       & (\textit{ii})  & 11.2857 & 18.4029 & 15.8846 & 15.1910 \\
          & (\textit{iii}) & 0.9404 & 0.9449 & 0.9528 & 0.9460 &       & (\textit{iii}) & 1.3318 & 4.7873 & 2.3864 & 2.8351 \\
          & (\textit{iv})  & 0.7630 & 0.7915 & 0.1055 & 0.5533 &       & (\textit{iv})  & 3.5707 & 11.5670 & 1196.7238 & 403.9538 \\
          & (\textit{v\textsubscript{1}}) & 0.9688 & 0.9658 & 0.9750 & 0.9699 &       & (\textit{v\textsubscript{1}}) & 1.0651 & 4.4800 & 1.3666 & 2.3039 \\
          & (\textit{v\textsubscript{2}}) & 0.9790 & 0.9737 & 0.9819 & 0.9782 &       & (\textit{v\textsubscript{2}}) & 1.0228 & 4.3977 & 1.2257 & 2.2154 \\
          & (\textit{v\textsubscript{3}}) & 0.9758 & 0.9777 & 0.9806 & 0.9780 &       & (\textit{v\textsubscript{3}}) & 1.0320 & 4.3235 & 1.2641 & 2.2065 \\
          \hline
    \multirow{7}[0]{*}{Kendall} & (\textit{i})   & 0.9945 & 0.9964 & 0.9966 & 0.9959 & \multirow{7}[0]{*}{MSD~\footnote{Mean signed difference}} & (\textit{i})   & -0.3789 & -2.7731 & -0.3010 & -1.1510 \\
          & (\textit{ii})  & 0.6113 & 0.6123 & 0.6091 & 0.6109 &       & (\textit{ii})  & 5.0912 & -2.2049 & 2.1883 & 1.6915 \\
          & (\textit{iii}) & 0.9946 & 0.9965 & 0.9969 & 0.9960 &       & (\textit{iii}) & -0.7558 & -3.9770 & -0.1579 & -1.6303 \\
          & (\textit{iv})  & 0.7014 & 0.6999 & 0.5554 & 0.6522 &       & (\textit{iv})  & -2.3183 & -8.0147 & 1194.9248 & 394.8640 \\
          & (\textit{v\textsubscript{1}}) & 0.9950 & 0.9966 & 0.9970 & 0.9962 &       & (\textit{v\textsubscript{1}}) & -0.9081 & -4.1164 & -0.9186 & -1.9810 \\
          & (\textit{v\textsubscript{2}}) & 0.9952 & 0.9967 & 0.9971 & 0.9963 &       & (\textit{v\textsubscript{2}}) & -0.8678 & -4.0628 & -0.8876 & -1.9394 \\
          & (\textit{v\textsubscript{3}}) & 0.9952 & 0.9967 & 0.9971 & 0.9963 &       & (\textit{v\textsubscript{3}}) & -0.8894 & -4.0389 & -0.8936 & -1.9406 \\
          \hline
    \end{tabular}%
    }
  \label{tab:results_MonteCarlo}%
\end{minipage}
\end{table}%

\subsection{Efficiency as a stochastic variable}\label{subsec:Efficiency as a stochastic variable}
\noindent Because of \gls{IKD}, the efficiency cannot be deterministic. Rather, it is necessarily stochastic. Let us assume that we can fit our estimates of $D^{k}$ to a probability density function from a well-known family of densities, \textit{e.g.}, Gaussian, Weibull, and \textit{t}-location scale, to name a few. From a finite set of possible functions, we may sort them (concerning their capacity of fitting to the empirical data) using the Bayesian information criterion, the Akaike information criterion, and the log-likelihood. The goodness-of-fit can be easily tested using the non-parametric Kolmogorov-Smirnov test, for instance. Let $f_{D^k} (\mathscr{d})$ be the probability density function associated with the (stochastic) distance of \gls{DMU} $k$ to the frontier $F$. The expected value of $D^k$ is as follows:

\begin{definition}[Expected value of $D^k$]\normalfont If $\Omega^k$ is the domain of $D^k$, \textit{i.e.}, $\Omega^k=[\min_\ell D^{k(\ell)}, \max_\ell D^{k(\ell)}]$, the expected value of these distances to the frontier is as follows:
\begin{equation}\label{eq:exp_D^k}
    \mathbb{E}(D^k) = \int_{\Omega^k} \mathscr{d}~ f_{D^k} (\mathscr{d}) ~ d\mathscr{d} .
\end{equation}

\end{definition}

\begin{example}[When $D^k\sim \text{Beta}(\alpha,\beta;q_1,q_2)$ distribution with shape parameters $\alpha,~\beta$]\label{example:5.1 Beta}\normalfont Suppose that the distance of \gls{DMU} $k$ to the frontier is a stochastic variable following a Beta distribution. Note that the Beta distribution considered here has support in $\mathscr{d}\in[q_1,q_2]$. In that case, should $\alpha$ and $\beta$ be two non-negative shape parameters, then \citep{Nadarajah2006}:
\begin{equation}
    f_{D^k} (\mathscr{d}) = \frac{1}{(q_2 - q_1)B(\alpha,\beta)}\left( \frac{\mathscr{d} - q_1}{q_2 - q_1}  \right)^{\alpha-1} \left( 1 -  \frac{\mathscr{d} - q_1}{q_2 - q_1}  \right)^{\beta-1},
\end{equation}
\noindent being $B(\alpha,\beta)=\frac{\Gamma(\alpha)\Gamma(\beta)}{\Gamma(\alpha + \beta)}$ the Euler's Beta function and $\Gamma$ the gamma function. The indefinite integral of $\mathscr{d}~ f_{D^k} (\mathscr{d})$ as in Equation (\ref{eq:exp_D^k}) is as follows:

\begin{equation}
\begin{split}
    I(\mathscr{d},q_1,q_2;\alpha,\beta) =& \int \mathscr{d}~ f_{D^k} (\mathscr{d}) ~ d\mathscr{d} =\\ =&\frac{\displaystyle\left( \frac{q_1 - \mathscr{d}}{q_1 - q_2} \right)^\alpha}{\alpha\beta B(\alpha,\beta)}  \left[
    (\beta q_1 + \alpha q_2)~{}_2 F_1 \left(\alpha;-\beta;1+\alpha;\frac{q_1 - \mathscr{d}}{q_1 - q_2}\right) - \alpha q_2 \left( \frac{\mathscr{d} - q_2}{q_1 - q_2} \right)^\beta \right],
\end{split}
\end{equation}

\noindent being ${}_2 F_1$ the Gauss hypergeometric function. Therefore, the expected value of $D^k$ is simply:

\begin{equation}
    \mathbb{E}(D^k) = I\left(\max_\ell D^{k(\ell)},q_1,q_2;\alpha,\beta\right) - I\left(\min_\ell D^{k(\ell)},q_1,q_2;\alpha,\beta\right),
\end{equation}

\noindent as a result of the Newton-Leibniz axiom (\textit{i.e.}, the second fundamental theorem of calculus). The definite integral in the domain of the Beta distribution gives us the following output:
\begin{equation}
    \mathbb{E}(D^k) = \int_{q_1}^{q_2} \mathscr{d}~ f_{D^k} (\mathscr{d}) ~ d\mathscr{d} = \frac{\beta q_1 + \alpha q_2}{B(\alpha,\beta)} \frac{\Gamma(\alpha)\Gamma(\beta)}{\Gamma(1+\alpha+\beta)}.
\end{equation}

Naturally, if $q_1 = 0,~q_2 = 1$, $\max_\ell D^{k(\ell)}=1$, and $\min_\ell D^{k(\ell)}=0$, we have $\mathbb{E}(D^k) = \frac{\alpha}{\alpha + \beta}$, which is the well-known expected value of a Beta distribution with shape parameters $\alpha,\beta$.
\end{example}

\begin{remark}\label{remark:Beta_to_Gauss}\normalfont
The Taylor's expansion of the Beta distribution shows that, for sufficiently large shape parameters such that $(\alpha + 1)/(\alpha -1 )\approx 1$ and $(\beta + 1)/(\beta -1 )\approx 1$, then the Beta distribution can be approximated by a Gaussian distribution, $\mathcal{N}$, with mean $\rho$ and standard deviation $\sigma$ \citep{Peizer68_1}:
\begin{equation}
    X\sim \text{Beta}(\alpha,\beta);~Y\sim \mathcal{N}(\rho, \sigma)~\text{and}~\alpha,\beta\gg 0 \Longrightarrow X \approx Y,
\end{equation}
\noindent such that:
\begin{equation}\label{eq:Beta_to_Normal}
    \begin{cases}
    \rho = \displaystyle\frac{\alpha}{\alpha+\beta},~\text{and}\\
    \sigma = \displaystyle\sqrt{\frac{\alpha\beta}{(\alpha+\beta)^2 (\alpha+\beta +1)}}.
    \end{cases}
\end{equation}

Of course, the expected value of the Gaussian distribution is equal to $\rho$, which should take the form mentioned in Equation (\ref{eq:Beta_to_Normal}).
\end{remark}

\vspace{.25cm}

Sometimes an efficiency score, rather than the distance to the frontier, is desirable. If the distance $D^k$ would be deterministic, we could derive a deterministic efficiency score, $\theta^k$, as follows:

\begin{definition}[(Deterministic) Efficiency score, $\theta^k$]\normalfont We may define efficiency as the relationship between the consumed inputs and the produced outputs of a \gls{DMU}. More precisely, efficiency measures the relationship between (optimal) targets and observations. Mathematically, the (deterministic) efficiency score of \gls{DMU} $k$ can be estimated as in \cite{Portela06}, \cite{Ferreira2016c}, and \cite{Ferreira2017b}:

\begin{equation}\label{eq:theta}
    \theta^k = \frac{\left( \displaystyle \prod_{i=1}^m \frac{(x_i^k)^\star}{x_i^k} \right)^{\frac{1}{m}}}{\left( \displaystyle \prod_{r=1}^s \frac{(y_r^k)^\star}{y_r^k} \right)^{\frac{1}{s}}}
    =\frac{\left( \displaystyle \prod_{i=1}^m \frac{x_i^k - \delta_i^x\cdot  D^k}{x_i^k} \right)^{\frac{1}{m}}}{\left( \displaystyle \prod_{r=1}^s \frac{y_r^k + \delta_r^y\cdot  D^k}{y_r^k} \right)^{\frac{1}{s}}},
\end{equation}

\noindent where the numerators $(x_i^k)^\star = x_i^k - \delta_i^x\cdot  D^k ~ (\leqslant x_i^k)$ and $(y_r^k)^\star = y_r^k + \delta_r^y\cdot  D^k ~ (\geqslant y_r^k)$, for $i=1,\ldots,m$ and $r=1,\ldots,s$, respectively, are the targets of inputs and outputs.
\end{definition}

\begin{remark}\normalfont
In line with \cite{Chambers96, Chambers98}, it usual to fix $(\delta_i^x,\delta_r^y,\delta_h^u)=(x_i^k,y_r^k,u_h^k)$ for any $i,r,h$. In that case, Equation (\ref{eq:theta}) simplifies to
\begin{equation}\label{eq:theta_simpl}
    \theta^k =  \frac{1-D^k}{1+D^k}.
\end{equation}
\end{remark}

From the simulations of the integrated approach, we cannot estimate the expected value of efficiency, $\mathbb{E}(\Theta^k)$, by simply replacing $D^k$ by $\mathbb{E}(D^k)$ in Equations (\ref{eq:theta}-\ref{eq:theta_simpl}). It is because $D^k$ is a stochastic variable with density $f_{D^k}$. We detail below how to estimate $\mathbb{E}(\Theta^k)$ based on the density of $\Theta^k = \frac{1-D^k}{1+D^k}\sim f_{\Theta^k} (\theta^k)$, which should depend on $f_{D^k}$.

If $W = X + Y$ is a random continuous variable resulting from the summation of two random continuous and independent variables, $X$ and $Y$, its probability density function is the convolution of the two densities:\footnote{Here and in the next few Equations, $X,~Y,$ and $U$ have nothing to do with the matrices of inputs, desirable outputs, and undesirable outputs, \textit{i.e.}, the observations plagued with \gls{IKD}. They are just random variables with densities.}

\begin{equation}
    f_W (w) = f_X (x)\ast f_Y (y) = \int_{-\infty}^{+\infty} f_X (x) f_Y (w -x)~dx,
\end{equation}

\noindent where $\ast$ denotes convolution. If $X$ is a constant (equal to 1 in the present case), its density is the Dirac's delta function:

\begin{equation}
    f_X (x-1) = \begin{cases}
    +\infty,~\text{if}~x=1\\
    0,~~~~~\text{otherwise.}
    \end{cases}
\end{equation}

\noindent In this case, the convolution between $f_Y (y)$ and $f_X (x-1)$ is $f_Y (y - 1)$ \citep{Schwartz1950, Schwartz1951}. In the same vein, if $Z = X - Y$ is a random continuous variable resulting from the subtraction of two random continuous and independent variables, we have:

\begin{equation}\label{eq:Z=X-Y pdf}
    f_Z (z) = f_X (x) \ast f_{-Y} (y) = \int_{-\infty}^{+\infty} f_X (x) f_{-Y} (z -x)~dx = \int_{-\infty}^{+\infty} f_X (x) f_{Y} (x - z)~dx.
\end{equation}

\noindent Being $X=1$ a constant, the convolution between $f_{-Y} (y)$ and $f_X (x - 1)$ is $f_{-Y} (y - 1) = f_Y (1 - y)$. From Equation (\ref{eq:theta_simpl}), we want to estimate the density of the random variable $U = Z(X,Y)/W(X,Y)$, such that $X \equiv 1$ and $Y \equiv D^k$. Following \cite{Fieller1932} and \cite{curtiss1941}, if $Z$ and $W$ are independent, then $U=Z/W \sim f_U (u)$, with:

\begin{equation}\label{eq:pdf_theta}
    f_U (u) = \int_{-\infty}^{+\infty} |z| f_W (uz) f_Z (z) ~dz = \int_{-\infty}^{+\infty} |z| f_Y (1 - uz) f_Y (z-1) ~dz.
\end{equation}

Because of Equations (\ref{eq:theta_simpl}) and (\ref{eq:pdf_theta}), the following Definition holds:

\begin{definition}[Probability density function of $\Theta^k$]\normalfont The density of $\Theta^k = \frac{1-D^k}{1+D^k}$ depends on the density of $D^k$ as follows:

\begin{equation}\label{eq:pdf_theta2}
    f_{\Theta^k} (\theta) = \int_{-\infty}^{+\infty} |\eta|~ f_{D^k} (1 - \theta \eta) ~ f_{D^k} (\eta - 1) ~d\eta = \int_0^{+\infty} \eta~ f_{D^k} (1 - \theta \eta) ~ f_{D^k} (\eta - 1) ~d\eta -\int_{-\infty}^0 \eta ~ f_{D^k} (1 - \theta \eta) ~ f_{D^k} (\eta - 1) ~d\eta.
\end{equation}

\end{definition}

\begin{example}[When $D^k\sim \text{Beta}(\alpha,\beta;0,1)$]\normalfont 
Suppose that $D^k\sim \text{Beta}(\alpha,\beta;0,1)$ distribution with shape parameters $\alpha,~\beta$. The mirror-image symmetry of the Beta distribution allows us to conclude that $D^k\sim \text{Beta}(\alpha,\beta;0,1) \Longleftrightarrow 1 - D^k = X \sim \text{Beta}(\beta,\alpha;0,1)$. Furthermore, we know from convolution with a Dirac's delta that $1+D^k = Y \sim f_Y (y)$, which has support in $- 1 \leqslant y \leqslant 0$. Therefore, $f_{\Theta^k} (\theta) = \int_{-\infty}^{+\infty} |y| f_X (\theta y) f_Y (y) dy = \int_{-1}^0 |y| f_X (\theta y) f_Y (y) dy$, with $f_X(\theta y) = \frac{\theta^{\beta - 1} y^{\beta - 1} (1-\theta y)^{\alpha - 1}}{B(\alpha,\beta)}$ and $f_Y (y) = \frac{y^{\beta-1} (1+y)^{\alpha-1}}{B(\alpha,\beta)}$. We finally get the following expression for $f_{\Theta^k} (\theta)$:\footnote{We used the software Wolfram Mathematica 12.1 (\url{https://www.wolfram.com/mathematica/}) to integrate.}

\begin{equation}\label{f_t_B}
    f_{\Theta^k} (\theta) = - (-1)^{1+2\beta} \frac{\theta^{\beta-1}}{B^2 (\alpha,\beta)} \frac{\Gamma(\alpha)\Gamma(2\beta)}{\Gamma(\alpha+2\beta)} {}_2 F_1 \left( 1-\alpha; 2\beta; \alpha+2\beta; -\theta  \right).
\end{equation}

\end{example}

\begin{example}[When $D^k\sim \mathcal{N}(\rho,\sigma)$, a Gaussian distribution with parameters $\rho,~\sigma$]\label{example:5.2 Beta}\normalfont Suppose that $D^k \sim \mathcal{N}(\rho,\sigma)$, where $\rho$ and $\sigma$ are, respectively, the mean and the standard deviation of the Gaussian distribution, $\mathcal{N}$. Then, $1-D^k \sim \mathcal{N}(1-\rho,\sigma)$ and $1+D^k \sim \mathcal{N}(1+\rho,\sigma)$. Following \cite{PhamGia2006}, we may use Hermite functions to solve Equation (\ref{eq:pdf_theta2}) and get $(1-D^k)/(1+D^k) \sim f_{\Theta^k} (\theta)$, with:

\begin{equation}\label{eq:f_t_N}
    f_{\Theta^k} (\theta) = \frac{T}{1 + \theta^2} ~{}_1 F_1\left(1;\frac{1}{2};\omega(\theta)\right),    
\end{equation}

\noindent being ${}_1 F_1$ the Kummer's classical confluent hypergeometric function of first kind \citep{CAMPOS2001177}, and:

\begin{equation}
    \begin{cases}
    T=\displaystyle\frac{1}{\pi}\exp \left[ -\frac{(1-\rho)^2 + (1+ \rho)^2}{2\sigma^2} \right],\\
    \text{and}\\
    \omega (\theta) = \displaystyle\frac{(1 + \rho)^2 + (1-\rho)^2\theta^2 + 2(1-\rho)(1+\rho)\theta}{2\sigma^2 (1 + \theta^2)}.
    \end{cases}
\end{equation}

Suppose that $\rho$ and $\sigma$ are defined as functions of $\alpha$ and $\beta$, shape parameters of a Beta distribution, according to the Equation (\ref{eq:Beta_to_Normal}). Because of Remark \ref{remark:Beta_to_Gauss}, Equations (\ref{f_t_B}) and (\ref{eq:f_t_N}) return similar values should $\alpha,\beta\gg 1$.

\end{example}

\vspace{.25cm}

As in the case of $D^k$, we usually need the expected value of the efficiency, $\mathbb{E}(\Theta^k)$. The next definition provides the equation for the expected value of $\Theta^k$, which also depends on the density of $D^k$:

\begin{definition}[Expected value of $\Theta^k$]\normalfont The expected value of $\Theta^k$ in the domain $\theta\in]0,1]$ is:

\begin{equation}\label{eq:exp_theta}
    \mathbb{E}(\Theta^k) = \int_{-\infty}^{+\infty} \theta f_{\Theta^k} (\theta)~ d\theta = \int_0^1 \theta f_{\Theta^k} (\theta)~ d\theta = \int_0^1 \int_{-\infty}^{+\infty} \theta |\varphi| f_{D^k} (1 - \theta \varphi) f_{D^k} (\varphi - 1)~ d\varphi ~d\theta.
\end{equation}

\end{definition}

\begin{remark}\normalfont
Let $\rho=\mathbb{E}(D^k)$. Equation (\ref{eq:exp_theta}) does not return $(1 - \rho)/(1+\rho)$ that would result from replacing $D^k$ by its expected value in Equation (\ref{eq:theta_simpl}). Indeed, following \cite{Stuart2010}, we may use the second order Taylor expansion over the ratio $(1-D^k)/(1+D^k)$ in the point $(1-\rho,1+\rho)$ to approximate the expected value $\mathbb{E}(\Theta^k)$ of Equation (\ref{eq:exp_theta}), as follows:

\begin{equation}\label{eq:Exp_theta2}
    \mathbb{E}(\Theta^k) = \frac{1-\rho}{1+\rho} - \frac{\text{Cov}(1-D^k,1+D^k)}{(1+\rho)^2} + \frac{(1-\rho)\text{Var}(1+D^k)}{(1+\rho)^3}.
\end{equation}

\noindent where $\text{Cov}$ is the covariance between two random variables, and $\text{Var}$ is the variance. Naturally, $\text{Cov}(1-D^k,1+D^k) = \text{Cov}(-D^k,D^k)=-\text{Cov}(D^k,D^k)=-\text{Var}(D^k)$. Also, let us use $\sigma^2$ to denote $\text{Var}(1+D^k)$ or $\text{Var}(D^k)$, interchangeably. Therefore, Equation (\ref{eq:Exp_theta2}) becomes:
\begin{equation}\label{eq:Exp_theta3}
    \mathbb{E}(\Theta^k) = \frac{1-\rho}{1+\rho} + \left[\frac{1}{(1+\rho)^2} + \frac{1-\rho}{(1+\rho)^3}\right] \sigma^2,
\end{equation}

Suppose that $\rho=\sigma=0$. In that case, $\mathbb{E}(\Theta^k) = 1$, meaning that the \gls{DMU} $k$ is efficient.
\end{remark}


\subsection{Additional definitions regarding the robustness analysis of the integrated approach}\label{subsec:Further considerations regarding the robustness analysis of the integrated approach}
\noindent Some additional definitions and remarks concerning the integrated approach's robustness analysis can be obtained from the iterating-based efficiency estimates. Based on the works of \cite{TERVONEN2007500}, \cite{TERVONEN2008}, and \cite{KADZINSKI20171}, let us consider the following three definitions:

\begin{definition}[Bucket, $b_g,~g=0,1,\ldots,G$]\label{def:bucket}\normalfont The bucket $b_g$ for $g=0,1,\ldots,G$ is an interval of efficiency (or distance $D^{\ast j},~j\in J,$ to the frontier, $F$) such that $b_g = ]\underline{b}_g,\overline{b}_g],$ $b_g \cap b_{g+1} = \emptyset$ and $\overline{b}_g - \underline{b}_g = \overline{b}_{g+1} - \underline{b}_{g+1}$ for any $g=0,1,\ldots,G-1$, and $\bigcup_{g=0}^G b_g = \displaystyle\left[\min_{j\in J} D^{\ast j},\max_{j\in J} D^{\ast j}\right]$.
\end{definition}

\begin{definition}[Efficiency bucket, $b_0$]\normalfont The efficiency bucket is defined by a single point in $\mathbb{R}$: $b_0 = \{0\}$.
\end{definition}

By Definition \ref{def:bucket}, $\bigcup_{g=1}^G b_g \triangleq \left[\min_{j\in J} D^{\ast j},\max_{j\in J} D^{\ast j}\right]$. If the frontier empirically estimated by \gls{DEA} envelops the entire sample, then $\min_{j\in J} D^{\ast j} = 0$, and the lower end of that range and $b_0$ overlap. If, due to IKD, some of the admissible points belonging to $\Lambda^j$ for some $j\in J$ fall outside the frontier, then $\min_{j\in J} D^{\ast j} < 0$ and \gls{DMU} $j$ is said to be super-efficient. Nonetheless, in the worst scenario, at least one of the remaining \glspl{DMU} in $J\setminus j$ is technically efficient (otherwise, we would have no frontier), and $\max_{j\in J} D^{\ast j} \geqslant 0$; thus, $b_0 \cap \bigcup_{g=0}^G b_g \neq \emptyset$ and $b_0 \subset \bigcup_{g=1}^G b_g$. Finally, if all observations are technically efficient, then $\min_{j\in J} D^{\ast j}=\max_{j\in J} D^{\ast j}$ for all $j\in J$, and $b_0 \equiv \bigcup_{g=1}^G b_g \Longrightarrow b_0 \subseteq \bigcup_{g=0}^G b_g$.

\begin{remark}\normalfont
The efficiency bucket verifies the condition $b_0 \subseteq \bigcup_{g=0}^G b_g$.
\end{remark}

It is usually desirable to study the frequency (either relative or absolute) in which a \gls{DMU} $k$ verifies an efficiency score (or distance to the frontier) belonging to a certain bucket. Such a frequency gives us the probability of observing such a performance for that \gls{DMU} in the presence of \gls{IKD}: 

\begin{definition}[Efficiency robustness interval index of \gls{DMU} $k\in J$, $\text{ERII}_g^k$] \label{def:ERII}\normalfont The efficiency robustness interval index of \gls{DMU} $k$ measures the number of times (in relation to $t$) that this \gls{DMU} belongs to the bucket $b_g=]\underline{b}_g, \overline{b}_g]$: 
\begin{equation}\label{eq:ERII_def}
\text{ERII}_g^k = \text{Pr}(D^k \in b_g)=\text{Pr}(\underline{b}_g < D^{k} \leqslant \overline{b}_g)= \frac{1}{t}\sum_{\ell=1}^t \mathbb{I}(\underline{b}_g < D^{k(\ell)}\leqslant \overline{b}_g),~g=0,1,\ldots,G.
\end{equation}
\end{definition}

\begin{remark}\label{remark:nonNeg_ERII}\normalfont
Since the indicator function $\mathbb{I}$ either is equal to 0 or 1, it follows that $\text{ERII}_g^k \geqslant 0$ for any $g=0,1,\ldots,G$ and $k\in J$.
\end{remark}

\begin{remark}\normalfont
It is straightforward to conclude that $\sum_{g=0}^G \text{ERII}_g^k = 1$.
\end{remark}

\begin{remark}\label{rem:EDk approx ERII}\normalfont
It is easy to see that
\begin{equation}\label{eq:EDk approx ERII}
    \mathbb{E}(D^k) \approx \sum_{g=0}^G \frac{\max b_g + \min b_g}{2}\cdot \text{ERII}_g^k = 0\cdot \text{ERII}_0^k + \sum_{g=1}^G \frac{\max b_g + \min b_g}{2}\cdot \text{ERII}_g^k = \sum_{g=1}^G \frac{\max b_g + \min b_g}{2}\cdot \text{ERII}_g^k,
\end{equation}
\noindent in which the approximation follows the Riemann-Darboux approach over the integration of Equation (\ref{eq:exp_D^k}), holding if the buckets' width is not too big. It follows that we can approximate $\mathbb{E}(\Theta^k)$ using $\text{ERII}_g^k$ and the buckets, as well.

If $\text{ERII}_g^k$ = 1 for $g=0$, then $\mathbb{E}(D^k) \approx 0$ by the Equation (\ref{eq:EDk approx ERII}).
\end{remark}

\begin{remark}\normalfont
Because of the non-negativity of $\text{ERII}_g^k$, see Remark \ref{remark:nonNeg_ERII}, it follows that $\mathbb{E}(D^k) \geqslant 0$. The lower bound, 0, means efficiency.
\end{remark}

Although useful in many situations, the scores $\mathbb{E}(D^k)$ and $\mathbb{E}(\theta^k)$ tell just a little about the robustness of efficiency estimates. It is common to associate the expected value with a confidence interval with a specified level. Let us consider only the confidence interval associated with the distance:

\begin{definition}[Confidence interval, $\Delta^k_\tau$]\normalfont Let $\{ D^{k[1]},\ldots, D^{k[\ell]},\ldots, D^{k[t]}\}$ be the list of sorted efficiency scores (or distances) achieved using $t$ loops, such that $D^{k[1]}\leqslant\cdots\leqslant D^{k[\ell]}\leqslant\cdots\leqslant D^{k[t]}$. The $\tau$-level confidence interval associated with $D^k$ is $\Delta^k_\tau = [D^{k[(1-\tau)t/2]};~D^{k[(1+\tau)t/2]}]$, for $0 < \tau < 1$.
\end{definition}

For instance, the 95\% confidence interval ($\tau=0.95$) associated with $D^k$, for $t=5,000$, is $\Delta_{0.95}^k = [D^{k[125]};~D^{k[4,875]}]$: 

\begin{equation}\label{eq:CI_95}
    \Delta_{0.95}^k = [LB_{95\%}^k, UB_{95\%}^k] = \begin{cases}
    LB_{95\%}^k = D^{k[125]}\\
    UB_{95\%}^k = D^{k[4,875]}
    \end{cases},
\end{equation}

\noindent where LB and UB stand, respectively, for the \textit{lower bound} and the \textit{upper bound} of the \gls{DMU} $k$'s confidence interval associated with the distance to the frontier.

Besides, it is usually useful to classify \glspl{DMU} based on their performance. We hereby classify them in the following categories:
\begin{enumerate}[label=\subscript{C}{{\arabic*}}:]
    \item Perfectly robust efficient;
    \item Sufficiently robust efficient;
    \item Neither robust efficient nor inefficient;
    \item Inefficient.
\end{enumerate}

If ``$A \succ B$'' denotes \textit{A is preferable to B}, then it is straightforward to conclude that $C_1 \succ C_2 \succ C_3 \succ C_4$. For this classification, it is sufficient to fix the confidence level $\tau\in]0,1[$, and know $\text{ERII}_g^k$ (for $g=0$), $\mathbb{E}(D^k)$ or $\mathbb{E}(\theta^k)$, and the confidence interval $\Delta^k_\tau$.

\begin{definition}[Perfectly robust efficient \gls{DMU}]\label{def:C1}\normalfont The \gls{DMU} $k$ is perfectly robust efficient if (and only if) $\mathbb{E}(D^k)=0$ or $\mathbb{E}(\theta^k)=1$. It is usually sufficient to have $\text{ERII}_g^k = 1$ for $g=0$, or $\Delta^k_\tau = \{ 0 \}$ for any $\tau\rightarrow 1$.
\end{definition}

\begin{definition}[Sufficiently robust efficient \gls{DMU}]\normalfont The \gls{DMU} $k$ is sufficiently robust efficient if $\tau \leqslant \text{ERII}_g^k < 1$, for $g=0$ and a predefined level $\tau \in ]0,1[$, which can be the same of the confidence interval.
\end{definition}

\begin{definition}[Neither robust efficient nor inefficient \gls{DMU}]\label{def:Neither robust efficient nor inefficient}\normalfont The \gls{DMU} $k$ is neither robust efficient nor inefficient if (and only if) both the following conditions are met: (\textit{i}) $\text{ERII}_g^k < \tau$ for $g = 0$; (\textit{ii}) $\Delta^k_\tau \supset \{ 0 \}$ and $\Delta^k_\tau \neq \{ 0 \}$. The second condition is equivalent of $LB_{95\%}\leqslant 0$ and $UB_{95\%} > 0$.
\end{definition}

\begin{definition}[Inefficient \gls{DMU}]\label{def:C4}\normalfont
To be inefficient, it is sufficient that \gls{DMU} $k$ verifies $LB_\tau^k > 0$ and $\mathbb{E}(D^k) > 0$.
\end{definition}

For a pair of \glspl{DMU}, $j$ and $k$, belonging to the same category, the next definitions hold:

\begin{definition}[Efficiency similarity between two \glspl{DMU} $j$ and $k$]\normalfont Suppose that a pair of \glspl{DMU}, $j$ and $k$, belong to the same category and $\Delta_\tau^j \cap \Delta_\tau^k \neq \emptyset$ (their confidence intervals are not disjoint). Then, there is no evidence that one \gls{DMU} outperforms the other. Symbolically, neither $j\succ k$ nor $k\succ j$.\label{def:tie}
\end{definition}

\begin{definition}[Outperformance]\normalfont Consider a pair of \glspl{DMU}, $j$ and $k$, belonging to the same category and such that their confidence intervals are disjoint: $\Delta_\tau^j \cap \Delta_\tau^k = \emptyset$. Thus, $j$ outperforms $k$ (\textit{i.e.}, $j\succ k$) if $UB_\tau^j < LB_\tau^k$.
\end{definition}


\subsection{Probability of a DMU outperforming other DMU}
\noindent We have previously shown how to check whether a \gls{DMU}, $j$, outperforms other \gls{DMU}, $k$, if they are both in the same efficiency category. It is based on the intersection of their confidence intervals associated with the distances to the frontier. In some empirical cases, though, it might be useful to determine the probability of \gls{DMU} $j$ outperforming \gls{DMU} $k$, \textit{i.e.}, $\text{Pr}(D^j < D^k)=\text{Pr}(D^j - D^k \leqslant 0)$, where $\leqslant$ holds if we assume that the difference $D^{j,k}=D^j - D^k$ is a continuous variable. 

We could simply determine the number of times that $D^j - D^k \leqslant 0$ in our simulation, obtaining:

\begin{equation}
    \text{Pr}(D^j < D^k) \approx \frac{1}{t}\sum_{\ell=1}^t \mathbb{I}(D^j - D^k \leqslant 0)= \frac{1}{t}\sum_{\ell=1}^t \mathbb{I}(D^{j,k} \leqslant 0),
\end{equation}

\noindent which is just an approximation of the \textit{true} probability because of the simulation process.

We prefer a more formal and elegant solution, which should account for the densities associated with each random variable related to the distance to the frontier. Let $f_{D^{j,k}}$ be the probability density function of $D^{j,k}$. Such a function is associated with the cumulative distribution function, $\mathscr{F}_{D}$, such that:

\begin{equation}\label{eq:cdf}
    \mathscr{F}_{D}(d) = \text{Pr}(D^{j,k}\leqslant d) = \int_{-\infty}^d f_{D^{j,k}}(w) ~dw.
\end{equation}

Since we need to estimate $\text{Pr}(D^j < D^k)=\text{Pr}(D^j - D^k \leqslant 0)=\text{Pr}(D^{j,k} \leqslant 0)$, we have: $\text{Pr}(D^j < D^k) = \mathscr{F}_{D}(0)$. 
The major difficulty here is that we do not know the function $f_{D^{j,k}}$, although we may know both $f_{D^k}$ and $f_{D^j}$, the probability density functions of $D^k$ and $D^j$, respectively. However, we may recall Equation (\ref{eq:Z=X-Y pdf}). Using the definition of cumulative distribution function in Equation (\ref{eq:cdf}), we get:

\begin{equation}\label{eq:cdf2}
    \mathscr{F}_{D} (d) = \text{Pr}(D^{j,k}\leqslant d) = \int_{-\infty}^d \int_{-\infty}^{+\infty} f_{D^j} (g) f_{D^k} (g - g')~dg~dg',
\end{equation}

\noindent for $D^{j,k} = D^j - D^k$. The Equation above can only be used if $D^j$ and $D^k$ are independent. If it is not the case, we have to consider the joint probability distribution, $f_{D^j D^k}(g,g')$, which may not be straightforward to get. Independence seems, however, a fair hypothesis. Finally, we obtain the value of $\text{Pr}(D^j < D^k)$ by simply replacing $d=0$ in Equation (\ref{eq:cdf2}).

\begin{example}[When the distances to the frontier follow (independent) Gaussian distributions]\label{example:D_Gauss_cdf}\normalfont
Suppose that $D^k \sim \mathcal{N}(\rho^k,\sigma^k)$ and $D^j \sim \mathcal{N}(\rho^j,\sigma^j)$, where $\rho$ and $\sigma$ denote the average and the standard deviation of the estimates $D^k$ and $D^j$, respectively. $\mathcal{N}$ is the Gaussian distribution, as usual. We assume that $D^j$ and $D^k$ are independent random variables. Then, $D^{j,k} = D^j - D^k \sim \mathcal{N}(\rho^j - \rho^k, \sqrt{(\sigma^j)^2 + (\sigma^k)^2})$. Using the Fourier transform, it is possible to show that Equation (\ref{eq:cdf2}) for $d=0$ reduces to $\mathscr{F}_{D} (0) = \text{Pr}(D^{j,k} \leqslant 0) = \Phi\left(\frac{-(\rho^j - \rho^k)}{\sqrt{(\sigma^j)^2 + (\sigma^k)^2}} \right)$, being $\Phi$ the cumulative distribution function of the standard Gaussian distribution. Of course, for other, more general, densities such a simplification is not that simple, requiring the double integration undertaken by the Equation (\ref{eq:cdf2}).
\end{example}

\subsection{Testing for differences in global results}
\noindent In many cases, researchers are interested on testing the influence of some parameters on efficiency, in a controlled way, for robustness purposes. In the present case, for example, we may test for the influence of the shape or the size of $\Lambda$ for modeling \gls{IKD} on efficiency. We recall that the convex polytopes defining $\Lambda$ depend on a set of parameters $a_{ql}$ and $b_q$, creating a pair $(a,b)$ that may influence $D^j$ for any $j\in J$. We denote $D^j(a,b)$ as the distance of \gls{DMU} $j$ as function of the pair $(a,b)$. Naturally, a change on one of these parameters may impact on $D^j$. Let us formulate the null hypothesis:

H\textsubscript{0}: $D^j(a,b)=D^j(a,b')$ for all $j\in J$,

\noindent in which $b'\neq b$. If the change of parameter $b$ to $b'$ does not produce effects over the distances to the frontier (and the frontier itself), then the null hypothesis is \textit{true} as there is no evidence to reject it for a given significance level, say 5\%. For any null hypothesis, there is an alternative against which the former is confronted. We may define this alternative as an inequality:

H\textsubscript{1}: $D^j(a,b)\neq D^j(a,b')$ for some $j\in J$.\footnote{Of course, other alternatives are possible, namely $D^j(a,b) > D^j(a,b')$ or $D^j(a,b) < D^j(a,b')$, turning the p-value computation slightly different from the one presented in this paper.}

Thus, if changes on the shape/size of $\Lambda$ impact on efficiency, we should reject H\textsubscript{0} in light of evidence, thus not rejecting H\textsubscript{1}.\footnote{Please, note that we do not accept an hypothesis because all depend on the significance level fixed and the simulation undertaken; the best we can say is that we do not reject it in light of the existing statistical evidence.}

A straightforward way of testing H\textsubscript{0} is using the p-value. We can take advantage of the $t$ estimates of distances per \gls{DMU}. Let us consider the following statistic:\footnote{Naturally, other statistics could be used instead; an example is the harmonic mean. In this case, the geometric mean is not advisable because distances are zero for efficient \glspl{DMU}.}

\begin{equation}
    T^{(\ell)} = \left(\frac{1}{n}\displaystyle\sum_{j=1}^n D^{j(\ell)}(a,b)\right)\bigg/\left(\frac{1}{n}\displaystyle\sum_{j=1}^n D^{j(\ell)}(a,b')\right),~\ell=1,\ldots,t.
\end{equation}

\noindent Note that other averages could be used rather than the simple mean of distances. Generically, except for the geometric mean (because distances can be zero), any H\"{o}lder mean might replace either the numerator, the denominator, or both in the previous Equation. Finally, since the alternative hypothesis, $H_1$, represents a difference, the p-value can be defined as follows:

\begin{equation}\label{eq:pvalue}
    \text{p-value} \approx \frac{2}{t}\min\left\{ \sum_{\ell=1}^t \mathbb{I}(T^{(\ell)}\leqslant 1),~ \sum_{\ell=1}^t \mathbb{I}(T^{(\ell)}\geqslant 1)\right\}.
\end{equation}

If the p-value is smaller than the fixed significance level (typically, 5\%), we reject the null hypothesis, meaning that the shape/size of $\Lambda$ play a meaningful role on efficiency assessment. In opposition, for p-value $\geqslant 5\%$, we do not have sufficient evidence to reject the null hypothesis, so $\Lambda$ has no impact on efficiency estimation.


\section{An empirical application}\label{sec:An empirical application}
\noindent This section presents an empirical application of the \gls{HR} procedure integrated with \gls{DEA} to assess a sample of 108 Portuguese public hospitals' technical efficiency. Model (\ref{eq:directionalDEA}) could be used to estimate their efficiencies, if no undesirable outputs had to be considered. However, evaluating hospitals' performance often requires considering undesirable outputs resulting from the production process, which are sometimes unavoidable. The smaller the produced amount of this kind of output, the better the performance of the hospital. Some approaches, including transforing of these quantities into desirable outputs, have been proposed in the literature \citep{ZANELLA2015517}. Transforming outputs does not seem the right approach, at least in the present case, due to two reasons. First, no consensus exists on which kind of transformation should be used. Second, neither of the possible transformations should return quantities with meaning. Therefore, we apply the model in Equation (\ref{eq:directionalDEA_corrected}) that disregards the transformation of outputs for the efficiency estimation.

\subsection{Sample, inputs, and outputs}\label{subsec:Sample, inputs, and outputs}
\noindent This section illustrates the \gls{HR} procedure's utilization to estimate the efficiency of a sample of 27 Portuguese public hospitals, that provided consistent data for four consecutive years: 2013 to 2016. 

Hospital data are available in a database maintained by the Ministry of Health and the Portuguese Central Administration of Health Systems.\footnote{Database: \url{http://benchmarking.acss.min-saude.pt} [in Portuguese], accessed: September 20, 2020.} Even though Portugal is a small country, the database is very rich in terms of measured indicators. Thus, the choice of inputs and outputs must be careful and parsimonious. These dimensions should also explain the hospital activity. Figure \ref{fig:Inputs, outputs, and undesirable outputs.} presents the adopted inputs and outputs (both desirable and undesirable) to execute the integration of \gls{HR} with Model (\ref{eq:directionalDEA_corrected}).

\begin{figure}[htbp]
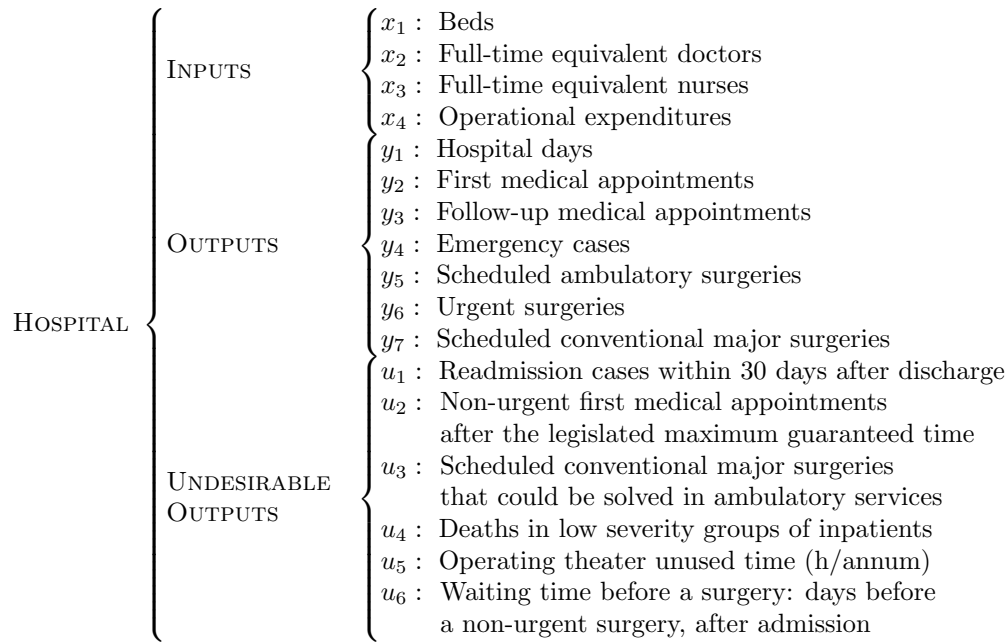

		\begin{equation*}
		\begin{split}
		&\text{\sc Hospital}
		\end{split}~
		\begin{cases}
		\textsc{Inputs}~& \begin{cases}
		\begin{split}
        & x_1:~\text{Beds}\\
        & x_2:~\text{Full-time equivalent doctors}\\
        & x_3:~\text{Full-time equivalent nurses}\\
        & x_4:~\text{Operational expenditures}
        \end{split}
		\end{cases}\\
        \textsc{Outputs}~& \begin{cases}
		\begin{split}
        & y_1:~\text{Hospital days}\\
		& y_2:~\text{First medical appointments}\\
        & y_3:~\text{Follow-up medical appointments}\\
        & y_4:~\text{Emergency cases}\\
        & y_5:~\text{Scheduled ambulatory surgeries}\\
        & y_6:~\text{Urgent surgeries}\\
        & y_7:~\text{Scheduled conventional major surgeries}
        \end{split}
		\end{cases}\\
        \begin{split} & \textsc{Undesirable} \\ & \textsc{Outputs}  \end{split}~& \begin{cases}
        \begin{split}
		& u_1:~\text{Readmission cases within 30 days after discharge}\\
        & \begin{split} u_2:&~\text{ Non-urgent first medical appointments}\\ & ~\text{ after the legislated maximum guaranteed time} \end{split}\\
        & \begin{split} u_3:&~\text{ Scheduled conventional major surgeries}\\ & ~\text{ that could be solved in ambulatory services} \end{split}\\
        & u_4:~\text{Deaths in low severity groups of inpatients} \\
        & u_5:~\text{Operating theater unused time (h/annum)} \\
        & \begin{split} u_6:&~\text{ Waiting time before a surgery: days before}\\ & ~\text{ a non-urgent surgery, after admission} \end{split}
        \end{split}
		\end{cases}
		\end{cases}
		\end{equation*}
		\caption{Inputs, outputs, and undesirable outputs. \label{fig:Inputs, outputs, and undesirable outputs.}}
	\end{figure}

Inputs characterize the resources consumed by hospitals to treat patients. Considering only the number of beds and the number of doctors and nurses is, in theory, insufficient to explain the profile of resource consumption. Hence, we include operational expenditures as extra input. It excludes the costs with staff to mitigate the problem of redundancy.

Outputs are measures of the hospital activity, regardless of the outcomes \citep{PEREIRA20201016}. This activity is related to the hospital's primary services: inpatients, medical appointments, emergency room, and operating theater (surgeries). As the complexity and severity of illness differ from patient to patient, each patient has a different cost of treatment. Hence, an adjustment mechanism is compulsory for the volume of services. In this case, we adopt the case-mix index \citep{MCRAE2020102086}. To compute the case-mix index, patients are clustered in diagnosis-related groups; patients belonging to a particular group of diagnoses are expected to represent similar hospital costs. The case-mix index associated with a hospital reflects the average cost of treating its patients relative to the national average unitary cost \citep{Chang2019}. Therefore, the larger the index, the more complex/severe patients are handled by the hospital. The Portuguese Ministry of Health considers three case-mix indices: inpatient, medical, and surgical specialties. Outputs were adjusted (weighted) accordingly to these indices, such that the volume of services became comparable among hospitals.

To evaluate the association between technical efficiency and quality, \cite{FERREIRA2018} and \cite{FerreiraITOR2020,FerreiraJPA2020} have considered dimensions such as: readmission rate within the first 30 days after discharge; non-urgent first medical appointments within the legislated maximum guaranteed time; outpatient (minor/ambulatory) surgeries on the potential outpatient procedures; the in-hospital death rate for low severity levels; operating theater capacity utilization; and waiting time before surgery, to name a few. Because of the convexity imposed by the fourth restriction of Model (\ref{eq:directionalDEA_corrected}), those rates cannot be considered as extra variables \citep{OLESEN2017640}. Thus, we have used the raw data associated with those ratios to define the undesirable outputs for this analysis. For instance, the in-hospital death rate is computed using the ratio number of deaths in hospital wards per 100 inpatients. In this case, we use the number of deaths as the undesirable output, considering only the less severe cases because of their lower mortality likelihood. Additionally, we considered: 

\begin {enumerate} [label=\itshape\alph*\upshape)]
    \item the \textit{number of readmissions within the first 30 days after discharge from the inpatient service} -- this dimension identifies lack of care appropriateness and, likely, patients' safety;
    
    \item the \textit{number of first medical appointments delayed beyond the legislated maximum guaranteed time}, which constitutes a barrier to access;
    
    \item the \textit{number of cases solved by major surgery but that could be solved using ambulatory (minor) procedure} -- ambulatory surgeries are less costly and more appropriate/safe to the patient requiring only minor procedures;
    
    \item the \textit{unused time of the operating theater}, which reveals non-optimization of resources usage; and
    
    \item the \textit{waiting time before surgery}, measured by the number of total days spent by patients waiting for a (non-urgent) surgery, once admitted to the hospital ward.
\end {enumerate} 

Using these undesirable outputs we may account for the preventable adverse events that result from unsafe and inappropriate care \citep{ROTHSCHILD200663, McCradden2020, Braspenning2020}, as well as the existence of some barriers to access. All raise costs to the health service \citep{Andel2012, Classen2011, Umscheid2011}.

As outputs, these undesirable quantities must be adjusted for the complexity/severity of illness of the patients treated by each hospital. However, we should point out that hospitals handling more complex patients are also more prone to observe adverse events (\textit{e.g.}, inpatients in severe conditions are more likely to decease than the others, less complicated). Hence, instead of multiply undesirable outputs by the case-mix index (which would lead to unfair comparisons), we divided them by it.

\subsection{Methodological considerations}\label{subsec:Methodological considerations}
\noindent We made two considerations to estimate the technical efficiency of Portuguese public hospitals:

\begin {enumerate} [label=\itshape\alph*\upshape)]
    \item Given the small size of Portugal, there is a limited number of public hospitals (our sample is composed of 27 hospitals). There is a dimensionality limitation associated with the sample's size, which means that small samples will likely result in low discrimination results. To mitigate this shortcoming of the linear model (\ref{eq:directionalDEA_corrected}), we pooled the data collected for the four years and constructed a common frontier, against which the efficiency of hospitals was estimated. It is equivalent to assume that no technological progress or regress was observed during that period. This assumption can be considered acceptable because the time lag is not considerable. The frontier is, thus, constructed using $27\times 4 = 108$ \glspl{DMU}.
    
    \item \gls{PCA} was used to narrow down the number of variables (one input, one output, and one undesirable output). It is a widely used technique to mitigate the problem of data redundancy. If two variables are highly and positively correlated, then using one instead of the other will not likely produce different outcomes (they are redundant). However, considering them could not be the right solution because of the dimensionality problem associated with non-parametric benchmarking models. \gls{PCA} reduces two or more correlated variables into a single one. It is what happened with our set of inputs, outputs, and undesirable outputs, as seventeen variables were reduced to three that can explain more than 90\% of the former ones' variability. 
\end {enumerate}

\subsection{On modeling the imperfect knowledge of data}\label{subsec:On modeling data imperfect knowledge}
\noindent We consider the following subscenario to model the \gls{IKD}: three quarters of the observations ($0.75\times 108 ~\text{DMUs} \times 3~ \text{variables} = 243$), randomly chosen, were replaced by sets to model IKD. These sets share the same geometry.
    
    

To define the sets $\Lambda$ to model IKD, we started by noticing the existence of three variables: $\tilde{x}=PCA(x),~\tilde{y}=PCA(y),$ and $\tilde{u}=PCA(u)$, all with unitary standard deviation. The utilization of PCA exacerbates or propagates the problem of IKD, justifying (even more) the adoption of the integrated approach HR+DEA.

We consider the superellipsoid defined below to characterize the sets used to model IKD. The superellipsoid generalizes the hyper-boxes, hyper-ellipsoids, and the hyper-rhombuses \citep{Lame1818}, as well as other geometric shapes. In our three-dimensional case (three variables after \gls{PCA}), the set $\Lambda$, centered on the empirical observations $(\tilde{x}^{j(0)},\tilde{y}^{j(0)},\tilde{u}^{j(0)}),~j\in J$, associated with the super-ellipsoid is as follows:

\begin{equation}\label{eq:set_lambda_spherical}
\Lambda^j = \left\{ (\tilde{x},\tilde{y},\tilde{u})\in\mathbb{R}_+^3~|~ \left(  \left| \frac{\tilde{x}-\tilde{x}^{j(0)}}{w_j^x} \right|^{O_1} + \left| \frac{\tilde{y}-\tilde{y}^{j(0)}}{w_j^y} \right|^{O_1}    \right)^{\frac{O_2}{O_1}} + \left| \frac{\tilde{u}-\tilde{u}^{j(0)}}{w_j^u} \right|^{O_2} \leqslant 1 \right\},~j\in J.
\end{equation}

\noindent Following \cite{Gielis2003}, in terms of spherical coordinates, the points of the boundary of $\Lambda$ in each \gls{HR} iteration $\ell$ can be parametrically defined by:

\begin{equation}\label{eq:3D_superellip}
\begin{cases}
\hat{x}^{j(\ell)}(\psi^{(\ell)},\phi^{(\ell)}) = \tilde{x}^{j(0)} + w_j^x g^{(1)}\left(\phi^{(\ell)},\frac{2}{O_2} \right)g^{(1)}\left(\psi^{(\ell)},\frac{2}{O_1} \right),\\
\hat{y}^{j(\ell)}(\psi^{(\ell)},\phi^{(\ell)}) = \tilde{y}^{j(0)} + w_j^y g^{(1)}\left(\phi^{(\ell)},\frac{2}{O_2} \right)g^{(2)}\left(\psi^{(\ell)},\frac{2}{O_1} \right),\\
\hat{u}^{j(\ell)}(\psi^{(\ell)},\phi^{(\ell)}) = \tilde{u}^{j(0)} + w_j^u g^{(2)}\left(\phi^{(\ell)},\frac{2}{O_1} \right),
\end{cases}~\ell=1,\ldots,t,~j\in J,
\end{equation}

\noindent where $g^{(1)}(a,b)=\text{sgn}(\cos a) |\cos a|^b$, $g^{(2)}(a,b)=\text{sgn}(\sin a) |\sin a|^b$, $\psi^{(\ell)}\in[-\pi,\pi]$, and $\phi^{(\ell)}\in[-\pi/2,\pi/2]$.\footnote{sgn(\textit{a}) denotes the sign function of \textit{a}; \textit{e.g.}, sgn(--5)=--1 and sgn(2)=1.} Because of the parametric definition of the hit points, instead of generating a directional vector $d^{(\ell)}$ as usual in \gls{HR}, we randomly select $\psi^{(\ell)} \sim \text{uniform}(-\pi,\pi)$ and $\phi^{(\ell)} \sim \text{uniform}(-\frac{\pi}{2},\frac{\pi}{2})$. In fact, this is equivalent of drawing $d^{(\ell)}_j$ because $\nabla^{j(\ell)} = (\hat{x}^{j(\ell-1)} - \tilde{x}^{j(\ell-1)},\hat{y}^{j(\ell-1)} - \tilde{y}^{j(\ell-1)},\hat{u}^{j(\ell-1)} - \tilde{u}^{j(\ell-1)})$ and $d^{(\ell)}_j/\| d^{(\ell)}_j \|_2=\nabla^{(\ell)}_j /\| \nabla^{(\ell)}_j \|_2$. Hence, $(\tilde{x}^{j(\ell)},\tilde{y}^{j(\ell)},\tilde{u}^{j(\ell)}) = (\tilde{x}^{j(\ell-1)},\tilde{y}^{j(\ell-1)},\tilde{u}^{j(\ell-1)}) + \lambda_j^{(\ell)} \xi_j^{(\ell)} \nabla^{(\ell)_j} $ is a random point within the super-ellipsoid defining the set $\Lambda$, following the HR procedure, being $\lambda_j^{(\ell)}=\| \nabla_j^{(\ell)} \|_2$, for any $\ell=1,\ldots,t$. Using spherical coordinates simplifies the \gls{HR} algorithm but the parametric definition of coordinates are not easily obtained for situations with more than three dimensions (three variables). Presently, we selected parameters $w_j=(w_j^x,w_j^y,w_j^u)=(0.2, 0.2, 0.2)$ to define the dimensions of each set.

We considered three different ways of modeling IKD in $\mathbb{R}^3_+$, as presented in Table \ref{tab:scenarios} and in Figure \ref{fig:superEllips2}. For instance, it is possible to show that, for significantly high orders $O_1$ and $O_2$ (\textit{i.e.}, $O_1,~O_2\rightarrow + \infty$), we get an hyper-box with volume $w^x\cdot w^y \cdot w^z$ from Equations (\ref{eq:set_lambda_spherical}) and (\ref{eq:3D_superellip}). 


\begin{table}[htbp]
  \centering
  \caption{Three subscenarios to model IKD.}
    \begin{tabular}{clcc}
    \hline
    Scenario & Case  & $O_1$    & $O_2$ \\ \hline
    (a)   & Hyper-box ($\Lambda^\text{HB}$) & $+\infty$ & $+\infty$ \\
    (b) & Hyper-ellipsoid ($\Lambda^\text{HE}$) & 2 & 2 \\
    (c) & Hyper-rhombus ($\Lambda^\text{HR}$) & 1 & 1\\ \hline
    \end{tabular}%
  \label{tab:scenarios}%
\end{table}


\subsection{Results}\label{subsec:Results}
\noindent Table \ref{tab:scenario_b_2013_2014} presents the distance and efficiency estimation using the integrated approach and considering the scenario (b) hyper-ellipsoid, and the years 2013 and 2014. Table B.3 (see Appendix B) presents the results of a similar exercise, this time regarding the years 2015 and 2016.\footnote{Appendix available at: \url{https://drive.google.com/drive/folders/1jAmKFzz_PWyPKSNTxqO_0mM8BWKWn3-D?usp=sharing}} Also in Appendix B, and in the same vein, Tables B.1 to B.5 exhibit the results of the same analysis, for two distinct scenarios: (a) hyper-box, and (c) hyper-rhombus. In these tables, we provide the efficiency robustness interval index, $ERII_0^k$, of each \gls{DMU} concerning the efficiency bucket, $b_0$, as detailed in Definition \ref{def:ERII} and Equation (\ref{eq:ERII_def}), to evaluate the probability of each \gls{DMU} be efficient, \textit{i.e.}, $ERII_0^k = \text{Pr}(D^k = 0)$. Such a chance is quite heterogeneous as the coefficient of variation is bigger than 150\% for most of scenarios and years. About the expected value of this probability, we may verify that it was about 20\% (regardless of the scenario adopted), decreasing to 4-10\% in the following three years (2014-2016). It suggests that the frontier constructed using the entire sample is mostly composed of observations from 2013.

The expected value of the distance to the frontier presented in the tables mentioned before was estimated using the approximation in Equation (\ref{eq:EDk approx ERII}), as described in Remark \ref{rem:EDk approx ERII}. In the present case, we divided the domain $[0,\max_{\ell} D^{j(\ell)}]$ into buckets of width 0.01, for all $j\in J=\{1,\ldots,108\}$. To complement this analysis, we constructed the 95\% confidence intervals following Equation (\ref{eq:CI_95}), and computed the empirical standard deviation, $\sqrt{Var(\Theta^k)}$, provided that $Var(\Theta^k) = \mathbb{E}((D^k)^2) - \mathbb{E}^2 (D^k)$ by definition. Also, we took advantage of these constructs to estimate the expected value of efficiency, $\mathbb{E}(\Theta^k)$, using Equation (\ref{eq:Exp_theta3}). Finally, we classified the \glspl{DMU} in four categories, $C_1$ (perfectly robust efficient) to $C_4$ (inefficient), following Definitions \ref{def:C1} to \ref{def:C4}, and assuming $\tau=0.95$.

In general, Portuguese public hospitals proved to be significantly and consistently inefficient across the period considered, a result of excessive resources consumption and high levels of undesirable outputs generated, given the level of desirable outputs delivered. To better understand it, we note that (in the absence of slacks) input targets are $(x_i^j)^\star = x_i^k (1 - D^k)$, while the (desirable) output targets are $(y_r^j)^\star = y_r^k (1 + D^k)$. Considering the scenario (a), on average, in 2013, the Portuguese public hospitals' consumption of inputs was 11.22\% above the optimal values, in 2013. In other words, input targets were about 89\% of the observed consumption profiles. In opposition, the target for the delivery of desirable outputs was 111.22\% of the observed values in the same year. Consistently throughout the three scenarios, it seems that performance decreased from 2013 to 2015, watching a tenuous improvement in the last biennium considered.

These results are in line with the decrease of \glspl{DMU} in category $C_1$, with the corresponding increasing in $C_4$. See Table \ref{tab:rate of DMUs per category} that shows the rate of hospitals per category, $C_1$ to $C_4$, year, and scenario. From 2014 onward, the number of observations considered inefficient remained nearly steady, because the lower bounds of the 95\% confidence intervals associated with $D^k$ were consistently larger than zero. Nonetheless, it seems that there is some dependence of the results on the scenario; \textit{e.g.}, 7\% of our \glspl{DMU} were perfectly robust efficient in 2016 in scenario (a), but none in that condition was observed in the same year in the other two scenarios. A plausible reason is because the domain of $\Lambda^{\text{HB}}$ is less restricted than either $\Lambda^\text{HE}$ or $\Lambda^\text{HR}$'s, for $w_j = (0.2, 0.2, 0.2)$. Besides, as shown in Table B.6 (in Appendix B), most observations, nearly 70\%, did never switch from one category to another. 

\begin{landscape}
\begin{table}[htbp]
  \centering
  \caption{Efficiency estimation using the Hit \& Run approach and the scenario (b). Results for 2013 and 2014.\label{tab:scenario_b_2013_2014}}
	\resizebox{0.7\textwidth}{!}{
    \begin{tabular}{lrrrrrrcrrrrrrc}
    \hline
          & \multicolumn{7}{c}{2013}                              & \multicolumn{7}{c}{2014} \\
          \hline
    Hospital & \multicolumn{1}{l}{$ERII_0^k$} & \multicolumn{1}{l}{$\mathbb{E}(D^k)$} & \multicolumn{1}{l}{$LB_{95\%}$} & \multicolumn{1}{l}{$UB_{95\%}$} & \multicolumn{1}{l}{$\sqrt{Var(D^k)}$} & \multicolumn{1}{l}{$\mathbb{E}(\Theta^k)$} & \multicolumn{1}{c}{Category} & \multicolumn{1}{l}{$ERII_0^k$} & \multicolumn{1}{l}{$\mathbb{E}(D^k)$} & \multicolumn{1}{l}{$LB_{95\%}$} & \multicolumn{1}{l}{$UB_{95\%}$} & \multicolumn{1}{l}{$\sqrt{Var(D^k)}$} & \multicolumn{1}{l}{$\mathbb{E}(\Theta^k)$} & \multicolumn{1}{c}{Category} \\
    \hline
    \multicolumn{1}{r}{1} & 0.00  & 0.2550 & 0.2425 & 0.2675 & 0.0066 & 0.5937 & \multicolumn{1}{l}{$C_4$} & 0.00  & 0.2916 & 0.2757 & 0.3089 & 0.0086 & 0.5485 & \multicolumn{1}{l}{$C_4$} \\
    \multicolumn{1}{r}{2} & 0.27  & 0.1735 & 0.0000 & 0.3304 & 0.1171 & 0.7213 & \multicolumn{1}{l}{$C_3$} & 0.00  & 0.3253 & 0.2770 & 0.4084 & 0.0377 & 0.5103 & \multicolumn{1}{l}{$C_4$} \\
    \multicolumn{1}{r}{3} & 0.00  & 0.0966 & 0.0936 & 0.0998 & 0.0017 & 0.8239 & \multicolumn{1}{l}{$C_4$} & 0.00  & 0.1775 & 0.1435 & 0.2108 & 0.0189 & 0.6989 & \multicolumn{1}{l}{$C_4$} \\
    \multicolumn{1}{r}{4} & 0.00  & 0.0456 & 0.0358 & 0.0573 & 0.0059 & 0.9129 & \multicolumn{1}{l}{$C_4$} & 0.00  & 0.1047 & 0.0939 & 0.1163 & 0.0061 & 0.8105 & \multicolumn{1}{l}{$C_4$} \\
    \multicolumn{1}{r}{5} & 0.00  & 0.2166 & 0.2100 & 0.2234 & 0.0036 & 0.6439 & \multicolumn{1}{l}{$C_4$} & 0.00  & 0.2541 & 0.2464 & 0.2620 & 0.0042 & 0.5948 & \multicolumn{1}{l}{$C_4$} \\
    \multicolumn{1}{r}{6} & 1.00  & 0.0000 & 0.0000 & 0.0000 & 0.0000 & 1.0000 & \multicolumn{1}{l}{$C_1$} & 0.00  & 0.0428 & 0.0305 & 0.0552 & 0.0072 & 0.9180 & \multicolumn{1}{l}{$C_4$} \\
    \multicolumn{1}{r}{7} & 0.00  & 0.1921 & 0.1847 & 0.2000 & 0.0041 & 0.6777 & \multicolumn{1}{l}{$C_4$} & 0.00  & 0.2339 & 0.2234 & 0.2452 & 0.0059 & 0.6209 & \multicolumn{1}{l}{$C_4$} \\
    \multicolumn{1}{r}{8} & 0.57  & 0.0124 & 0.0000 & 0.0463 & 0.0161 & 0.9761 & \multicolumn{1}{l}{$C_3$} & 0.93  & 0.0102 & 0.0000 & 0.1432 & 0.0360 & 0.9824 & \multicolumn{1}{l}{$C_3$} \\
    \multicolumn{1}{r}{9} & 0.00  & 0.1239 & 0.1185 & 0.1296 & 0.0030 & 0.7795 & \multicolumn{1}{l}{$C_4$} & 0.00  & 0.1621 & 0.1552 & 0.1695 & 0.0038 & 0.7210 & \multicolumn{1}{l}{$C_4$} \\
    \multicolumn{1}{r}{10} & 0.00  & 0.1022 & 0.0862 & 0.1271 & 0.0096 & 0.8147 & \multicolumn{1}{l}{$C_4$} & 0.00  & 0.1213 & 0.0296 & 0.2412 & 0.0671 & 0.7900 & \multicolumn{1}{l}{$C_4$} \\
    \multicolumn{1}{r}{11} & 0.00  & 0.0346 & 0.0238 & 0.0456 & 0.0063 & 0.9331 & \multicolumn{1}{l}{$C_4$} & 0.00  & 0.0664 & 0.0537 & 0.0792 & 0.0077 & 0.8755 & \multicolumn{1}{l}{$C_4$} \\
    \multicolumn{1}{r}{12} & 1.00  & 0.0000 & 0.0000 & 0.0000 & 0.0000 & 1.0000 & \multicolumn{1}{l}{$C_1$} & 0.00  & 0.0636 & 0.0609 & 0.0666 & 0.0016 & 0.8803 & \multicolumn{1}{l}{$C_4$} \\
    \multicolumn{1}{r}{13} & 0.00  & 0.0095 & 0.0084 & 0.0106 & 0.0006 & 0.9812 & \multicolumn{1}{l}{$C_4$} & 0.02  & 0.0353 & 0.0012 & 0.0692 & 0.0202 & 0.9325 & \multicolumn{1}{l}{$C_4$} \\
    \multicolumn{1}{r}{14} & 0.00  & 0.0585 & 0.0455 & 0.0727 & 0.0074 & 0.8895 & \multicolumn{1}{l}{$C_4$} & 0.00  & 0.1715 & 0.1573 & 0.1860 & 0.0079 & 0.7072 & \multicolumn{1}{l}{$C_4$} \\
    \multicolumn{1}{r}{15} & 0.00  & 0.2263 & 0.2167 & 0.2360 & 0.0053 & 0.6310 & \multicolumn{1}{l}{$C_4$} & 0.00  & 0.2650 & 0.2527 & 0.2774 & 0.0066 & 0.5811 & \multicolumn{1}{l}{$C_4$} \\
    \multicolumn{1}{r}{16} & 0.00  & 0.0089 & 0.0026 & 0.0395 & 0.0100 & 0.9826 & \multicolumn{1}{l}{$C_4$} & 0.47  & 0.0341 & 0.0000 & 0.1498 & 0.0470 & 0.9381 & \multicolumn{1}{l}{$C_3$} \\
    \multicolumn{1}{r}{17} & 1.00  & 0.0000 & 0.0000 & 0.0000 & 0.0000 & 1.0000 & \multicolumn{1}{l}{$C_1$} & 0.00  & 0.0452 & 0.0435 & 0.0469 & 0.0009 & 0.9134 & \multicolumn{1}{l}{$C_4$} \\
    \multicolumn{1}{r}{18} & 0.00  & 0.1392 & 0.1343 & 0.1442 & 0.0027 & 0.7556 & \multicolumn{1}{l}{$C_4$} & 0.00  & 0.1541 & 0.1317 & 0.1760 & 0.0129 & 0.7332 & \multicolumn{1}{l}{$C_4$} \\
    \multicolumn{1}{r}{19} & 0.00  & 0.2057 & 0.2001 & 0.2117 & 0.0032 & 0.6588 & \multicolumn{1}{l}{$C_4$} & 0.00  & 0.2168 & 0.2104 & 0.2236 & 0.0036 & 0.6437 & \multicolumn{1}{l}{$C_4$} \\
    \multicolumn{1}{r}{20} & 0.00  & 0.0798 & 0.0726 & 0.0869 & 0.0041 & 0.8523 & \multicolumn{1}{l}{$C_4$} & 0.00  & 0.1339 & 0.1149 & 0.1525 & 0.0109 & 0.7640 & \multicolumn{1}{l}{$C_4$} \\
    \multicolumn{1}{r}{21} & 0.00  & 0.1218 & 0.0914 & 0.1501 & 0.0178 & 0.7833 & \multicolumn{1}{l}{$C_4$} & 0.00  & 0.1857 & 0.1775 & 0.1943 & 0.0046 & 0.6867 & \multicolumn{1}{l}{$C_4$} \\
    \multicolumn{1}{r}{22} & 0.97  & 0.0013 & 0.0000 & 0.0150 & 0.0086 & 0.9975 & \multicolumn{1}{l}{$C_2$} & 0.86  & 0.0021 & 0.0000 & 0.0246 & 0.0061 & 0.9959 & \multicolumn{1}{l}{$C_3$} \\
    \multicolumn{1}{r}{23} & 0.00  & 0.2775 & 0.2634 & 0.2920 & 0.0076 & 0.5656 & \multicolumn{1}{l}{$C_4$} & 0.00  & 0.3161 & 0.2978 & 0.3330 & 0.0098 & 0.5197 & \multicolumn{1}{l}{$C_4$} \\
    \multicolumn{1}{r}{24} & 0.00  & 0.1982 & 0.1877 & 0.2085 & 0.0055 & 0.6692 & \multicolumn{1}{l}{$C_4$} & 0.00  & 0.2231 & 0.2107 & 0.2353 & 0.0065 & 0.6353 & \multicolumn{1}{l}{$C_4$} \\
    \multicolumn{1}{r}{25} & 0.00  & 0.0720 & 0.0668 & 0.0779 & 0.0030 & 0.8657 & \multicolumn{1}{l}{$C_4$} & 0.00  & 0.1445 & 0.1386 & 0.1504 & 0.0032 & 0.7475 & \multicolumn{1}{l}{$C_4$} \\
    \multicolumn{1}{r}{26} & 0.00  & 0.1803 & 0.1556 & 0.2044 & 0.0142 & 0.6948 & \multicolumn{1}{l}{$C_4$} & 0.00  & 0.1909 & 0.1662 & 0.2137 & 0.0144 & 0.6797 & \multicolumn{1}{l}{$C_4$} \\
    \multicolumn{1}{r}{27} & 0.84  & 0.0083 & 0.0000 & 0.0956 & 0.0238 & 0.9846 & \multicolumn{1}{l}{$C_3$} & 0.00  & 0.0144 & 0.0056 & 0.0895 & 0.0208 & 0.9725 & \multicolumn{1}{l}{$C_4$} \\
    \hline
    Mean  & 0.2094 & 0.1052 & 0.0904 & 0.1249 & 0.0107 & 0.8218 &       & 0.0842 & 0.1476 & 0.1296 & 0.1788 & 0.0141 & 0.7556 &  \\
    Std.Dev. & 0.3783 & 0.0883 & 0.0866 & 0.0951 & 0.0216 & 0.1422 &       & 0.2455 & 0.0962 & 0.0960 & 0.0928 & 0.0154 & 0.1471 &  \\
    CV    & 181\% & 84\%  & 96\%  & 76\%  & 203\% & 17\%  &       & 292\% & 65\%  & 74\%  & 52\%  & 109\% & 19\%  &  \\
    \hline
    \end{tabular}
		}
\end{table}%
\end{landscape}

\begin{table}[t]
  \centering
  \caption{Rate of hospitals per category, year, and scenario. \label{tab:rate of DMUs per category}}
    \begin{tabular}{lrrrrrrrrrrrr}
    \hline
    \multicolumn{1}{c}{\multirow{2}[0]{*}{Category}} & \multicolumn{4}{c}{Scenario (a)} & \multicolumn{4}{c}{Scenario (b)} & \multicolumn{4}{c}{Scenario (c)} \\
          & 2013  & 2014  & 2015  & 2016  & 2013  & 2014  & 2015  & 2016  & 2013  & 2014  & 2015  & 2016 \\ \hline
    $C_1$ & 15\%  & 4\%   & 0\%   & 7\%   & 11\%  & 0\%   & 0\%   & 7\%   & 15\%  & 4\%   & 0\%   & 0\% \\
    $C_2$ & 0\%   & 0\%   & 0\%   & 0\%   & 4\%   & 0\%   & 0\%   & 0\%   & 0\%   & 0\%   & 0\%   & 0\% \\
    $C_3$ & 7\%   & 0\%   & 0\%   & 4\%   & 11\%  & 11\%  & 7\%   & 4\%   & 19\%  & 11\%  & 11\%  & 11\% \\
    $C_4$ & 78\%  & 96\%  & 100\% & 89\%  & 74\%  & 89\%  & 93\%  & 89\%  & 67\%  & 85\%  & 89\%  & 89\% \\
    \hline
    \end{tabular}
\end{table}%

As mentioned before, there are suspicions that the model specifications might play a pivotal role on efficiency or distance estimation. A simple correlation analysis of the expected value of efficiency from the three scenarios states that all Pearson's correlation coefficients are above 0.92 and statistically significant. However, it only allows us to conclude that all scenarios tend to produce results in the same direction. It is not sufficient to conclude whether one model produces distinct outcomes or not, though. Table \ref{tab:compar_scenarios} compares the three scenarios using the p-value and Equation (\ref{eq:pvalue}) with three distinct H\"{o}lder orders to estimate the statistic $T^{\ell},~\ell=1,\ldots,t$: 1 (simple arithmetic mean of distances), 2 (root mean square), and $\infty$ (maximum). In this case, we are testing if the orders specified in Table \ref{tab:scenarios} should impact in distance estimates (and efficiency, accordingly). If the p-value is larger than 5\%, then we have no statistical evidence supporting the rejection of the null hypothesis of similar distributions. For instance, the p-value of the null hypothesis H\textsubscript{0}: $D^j (\text{scenario (b)}) = D^j (\text{scenario (c)})$ for most $j\in J$ is 0.8900. Although assuming the common arithmetic mean to test such an hypothesis would never result in the latter's rejection, the results in this table suggest the possibility of concluding that there is evidence of efficiency dependence on the model specifications. Nonetheless, such a dependence is not yet well understood, being left for further research.

\begin{table}[htbp]
    \centering
    \caption{Comparison of the three scenarios: p-values considering three orders for the H\"{o}lder mean - 1 (arithmetic mean), 2 (root mean square), and $\infty$ (maximum).}
    \label{tab:compar_scenarios}
    \begin{tabular}{ccc}
         (\textit{i}) Order 1 & (\textit{ii}) Order 2 & (\textit{iii}) Order $\infty$ \\
         
         \begin{tabular}{rrrr}
         \hline
          & \multicolumn{1}{r}{(a)} & \multicolumn{1}{r}{(b)} & \multicolumn{1}{r}{(c)} \\ \hline
    (a)   & 1     & 0.6920 & 0.6208 \\
    (b)   &       & 1     & 0.9512 \\
    (c)   &       &       & 1 \\ \hline
    \end{tabular}
         
         &
         
         \begin{tabular}{rrrr}
         \hline
          & \multicolumn{1}{r}{(a)} & \multicolumn{1}{r}{(b)} & \multicolumn{1}{r}{(c)} \\ \hline
    (a)   & 1     & 0.4844 & 0.5196 \\
    (b)   &       & 1     & 0.8900 \\
    (c)   &       &       & 1 \\ \hline
    \end{tabular}%

        &
        
        \begin{tabular}{rrrr} \hline
          & \multicolumn{1}{r}{(a)} & \multicolumn{1}{r}{(b)} & \multicolumn{1}{r}{(c)} \\ \hline
    (a)   & 1     & 0.0800  & 0.0000 \\
    (b)   &       & 1     & 0.9408 \\
    (c)   &       &       & 1 \\ \hline
    \end{tabular}
         
    \end{tabular}
\end{table}

To exemplify the stochastic nature of our estimates and the usefulness of outputs from Subsection \ref{subsec:Efficiency as a stochastic variable}, we take a look at hospital 1, which corresponds to the \glspl{DMU} 1, 28, 55, and 82, from 2013 to 2016. Regarding scenario (a), the distribution of distances was mostly composed of peaks, being impossible to detect an appropriate parametric function for the densities. In opposition, the distributions of those \glspl{DMU} in scenarios (b) and (c) seem very well behaved, and we could approximate the empirical distribution by Beta distributions. Details about this distribution can be found in Examples \ref{example:5.1 Beta} and \ref{example:5.2 Beta}. Figures B.4 and B.5 (Appendix B) provide the densities associated with the distance of hospital 1 to the frontier, as well as the best fit, which was always a Beta distribution. We used the \textsc{Matlab} (R2018b) function 'fitmethis' and the log-likelihood criterion to select the parametric distribution function with the \textit{best} goodness-of-fit.\footnote{Francisco de Castro (2020). fitmethis available at \url{https://www.mathworks.com/matlabcentral/fileexchange/40167-fitmethis}, MATLAB Central File Exchange. Retrieved December 12, 2020.} We tested the goodness-of-fit of the Beta distributions using the Kolmogorov-Smirnov non-parametric test and the p-values were always above 0.90. It means that the Beta distribution is appropriate to model the densities of the distance of hospital 1 to the efficiency frontier: $D^k \sim \text{Beta}(\alpha,\beta),~k=1,28,55,82$. We placed the four yearly distributions side-by-side in Figure \ref{fig:densities_b_c_DMU1}, and the corresponding shape parameters, $\alpha$ and $\beta$, in Table \ref{tab:parameters_Beta_Fig}(\textit{i}). 

As we can see, all those parameters $\alpha$ and $\beta$ are substantially larger than 1. Thus, $(\alpha+1)/(\alpha-1)\approx 1$ and $(\beta+1)/(\beta-1)\approx 1$. According to the Remark \ref{remark:Beta_to_Gauss}, in that case, we can approximate any Beta distribution to a Gaussian, after a suitable change of parameters to $\rho$ (location) and $\sigma$ (scale); see Equation (\ref{eq:Beta_to_Normal}). Gaussian distributions are much more manageable than the Beta ones. Table \ref{tab:parameters_Beta_Fig}(\textit{ii}) provides these parameters and the expected value of efficiency. Efficiency results after applying the Equation (\ref{eq:Beta_to_Normal}) to estimate $\rho = \mathbb{E}(D^k)$ and $\sigma=\sqrt{Var(D^k)}$, and Equation (\ref{eq:Exp_theta3}) to estimate $\mathbb{E}(\Theta^k)$. Note the difference in $\mathbb{E}(\Theta^k)$ for $k=82$ in scenario (b) and 2016 between Tables \ref{tab:parameters_Beta_Fig}(\textit{ii}) and \ref{tab:scenario_b_2015_2016}. Although they have been estimated using the same Equation, the differences lie in parameters $\rho$ and $\sigma$, which were estimated either directly from the 5,000 estimates or after their parameterization by a Beta (Gaussian) distribution.

\begin{figure}[t]
\centering
	\subfloat[Scenario (b)]{\includegraphics[width = .45\columnwidth]{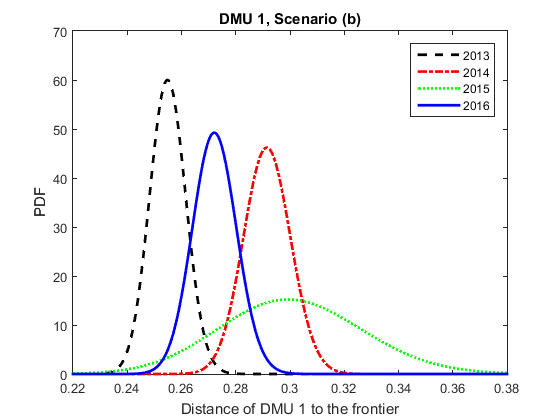}} 
	\subfloat[Scenario (c)]{\includegraphics[width = .45\columnwidth]{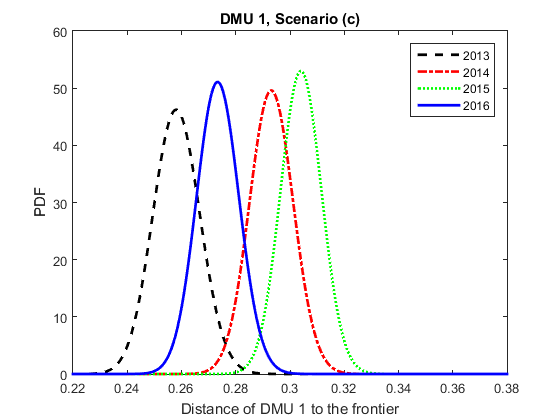}}
\caption{Beta distributions (side-by-side) of the distance of hospital 1 (DMUs 1, 28, 55, and 82) to the frontier, per year (2013-2016), and concerning the scenarios (b) hyper-ellipsoid and (c) hyper-rhombus.}
\label{fig:densities_b_c_DMU1}
\end{figure}

\begin{table}[t]
  \centering
  \caption{Parameters and efficiency results associated with the Beta distributions exhibited in Figure \ref{fig:densities_b_c_DMU1} (hospital 1, \textit{i.e.}, DMUs 1, 28, 55, and 82). \label{tab:parameters_Beta_Fig}}
\begin{tabular}{cc}
\centering
(\textit{i}) Shape parameters, $\alpha$ and $\beta$ & {(\textit{ii}) Efficiency results} \\
\begin{tabular}{lrrrr}
    \hline
          & \multicolumn{2}{c}{Scenario (b)} & \multicolumn{2}{c}{Scenario (c)} \\
          & \multicolumn{1}{r}{$\alpha$} & $\beta$ & \multicolumn{1}{r}{$\alpha$} & $\beta$ \\
    \hline
    2013  & \multicolumn{1}{l}{1,095.5} & 3,201.2 & 663.8 & 1,906.9 \\
    2014  & 807.4 & 1,961.5 & 939.3 & 2,264.2 \\
    2015  & 91.7  & \multicolumn{1}{r}{213.4} & \multicolumn{1}{l}{1,135.8} & 2,598.8 \\
    2016  & 821.6 & 2,196.5 & 890.4 & 2,365.7 \\
    \hline
    \end{tabular}
&
\begin{tabular}{rrrrrrr}
\hline
          & \multicolumn{3}{c}{Scenario (b)} & \multicolumn{3}{c}{Scenario (c)} \\
          & \multicolumn{1}{r}{$\rho$} & \multicolumn{1}{r}{$\sigma$} & \multicolumn{1}{r}{$\mathbb{E}(\Theta^k)$} & \multicolumn{1}{r}{$\rho$} & \multicolumn{1}{r}{$\sigma$} & \multicolumn{1}{r}{$\mathbb{E}(\Theta^k)$} \\
          \hline
    2013  & 0.2550 & 0.0066 & 0.5937 & 0.2582 & 0.0086 & 0.5896 \\
    2014  & 0.2916 & 0.0086 & 0.5485 & 0.2932 & 0.0080 & 0.5466 \\
    2015  & 0.3006 & 0.0262 & 0.5384 & 0.3041 & 0.0075 & 0.5336 \\
    2016  & 0.2722 & 0.0081 & 0.5721 & 0.2735 & 0.0078 & 0.5706 \\
    \hline
    \end{tabular}%

\end{tabular}
\end{table}%

As we can see in the densities representation, there was a significant shift to the right from 2013 to 2014, indicating a considerable efficiency worsening, followed by a smaller shift in the next biennium (2014-2015), suggesting a smoother worsening of efficiency compared to the previous downfall. Then, there was an improvement of efficiency, from 2015 to 2016; but the same hospital remained less efficient in 2016 than it was in 2013. The major difference between both plots is the variance of the distribution of \gls{DMU} 55 (2015) in scenario (b), which can be clearly noticed by the results displayed in Table \ref{tab:parameters_Beta_Fig}(\textit{ii}). Besides, it is interesting these results with the probability of one \gls{DMU} outperforming another. Table \ref{tab:prob_j_outp_k} highlights the cross-probabilities $\text{Pr}(D^j\leqslant D^k)$ using the cumulative distribution function of the standard Gaussian distribution (see Example \ref{example:D_Gauss_cdf}), a valid proxy for Beta distributions with large shape parameters. According to our results, it is very unlikely that \gls{DMU} 1 (2013) could be outperformed by any other observation of the same hospital from 2014 onward. Indeed, the probabiliy of this observation being outperformed is, at the most, $9.52\%$. In opposition, \glspl{DMU} 28 (2014) and 55 (2015) are almost certainly outperformed by the others. These results seem to be consistent with the intersection of densities of distance, as shown in Figure \ref{fig:densities_b_c_DMU1}.

\begin{table}[t]
  \centering
  \caption{Probability of DMU $j$ outperforming DMU $k$. \label{tab:prob_j_outp_k}}
    \begin{tabular}{rrrr}
    \hline
    \multicolumn{1}{l}{DMU $j$} & \multicolumn{1}{l}{DMU $k$} & \multicolumn{1}{l}{Scenario (b)} & \multicolumn{1}{l}{Scenario (c)} \\ \hline
    1 (2013)  & 28 (2014)  & 99.96\% & 99.85\% \\
    1 (2013)  & 55 (2015)  & 95.41\% & 100.00\% \\
    1 (2013)  & 82 (2016)  & 95.02\% & 90.48\% \\
    28 (2014)  & 55 (2015)  & 62.73\% & 83.92\% \\
    28 (2014)  & 82 (2016)  & 5.09\% & 3.90\% \\
    55 (2015)  & 82 (2016)  & 15.08\% & 0.23\% \\
    \hline
    \end{tabular}
\end{table}%

Rather than the distance to the frontier, researchers are often more interested in the efficiency score. Since the distance $D^k$ after the integrated approach becomes a stochastic variable, it is natural that $\Theta^k$ is also stochastic. Using either Equation (\ref{f_t_B}) or (\ref{eq:f_t_N}) with the parameters in Table \ref{tab:parameters_Beta_Fig}(\textit{i}), we generated the points of the distribution, which were then adjusted by a known parametric density function. In the present case, the efficiency of hospital 1, regardless of the year and scenario, (b) or (c), is well modelled by Generalized Extreme Value ($\mathcal{GEV}$) distributions with parameters $\rho$ (location), $\sigma$ (scale), and $\psi$ (shape): $\Theta^k \sim \mathcal{GEV}(\rho,\sigma,\psi),~k=1,28,55,82$. $\mathcal{GEV}$ distributions maximized the log-likelihood criterion among the parametric families of densities. Figure \ref{fig:densities_b_c_DMU1_eff} exhibits the densities per year, side-by-side. The behaviour of efficiency as a stochastic variable is consistent with the analysis of Figure \ref{fig:densities_b_c_DMU1}. Table \ref{tab:param_GEV} contains the parameters of the $\mathcal{GEV}$ distributions. The expected value of efficiency results from $\mathbb{E}(\Theta^k) = \rho + \sigma (g_1 - 1)/\psi$ for $g_1 = \Gamma(1-\psi)$, since $\psi < 0$. Likewise, $Var(\Theta^k) = \sigma^2 (g_2 - g_1^2)/\psi^2$ for $g_2 = \Gamma(1-2\psi)$. Note that the maximum difference between the expected values between Tables \ref{tab:parameters_Beta_Fig}(\textit{ii}) and \ref{tab:param_GEV} is $0.0002$, thus meaningless. It is worth of mentioning that any $\mathcal{GEV}$ distribution with negative shape parameter, $\psi<0$, belongs to the Weibull family of densities ($\mathcal{GEV}$ distributions type III; see \cite{BALI2003423}). More precisely, the stochastic variable $\Theta^k,~k=1,28,55,82,$ follows reverse Weibull distributions. Interestingly, \cite{Park99}, \cite{Cazals2002}, and \cite{Daouia2005, Daouia2007} have concluded that efficiency asymptotically tends to follow Weibull distributions. We leave for further research investigating whether this trend is observed for the remaining \glspl{DMU} in our sample and in more general cases.

\begin{figure}[t]
\centering
	\subfloat[Scenario (b)]{\includegraphics[width = .45\columnwidth]{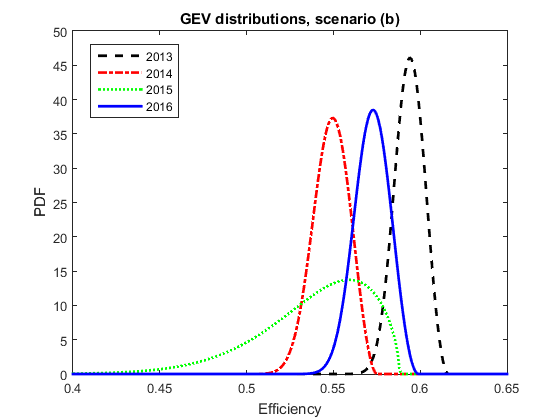}} 
	\subfloat[Scenario (c)]{\includegraphics[width = .45\columnwidth]{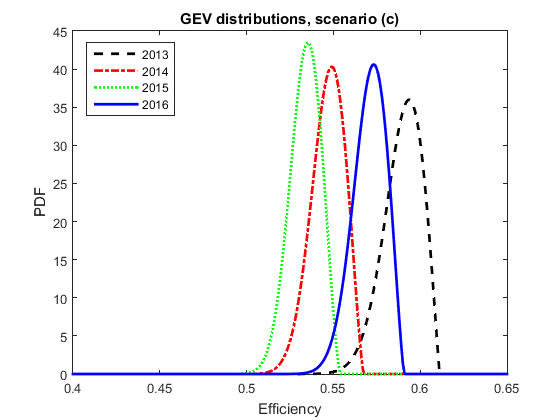}}
\caption{$\textsc{GEV}$ distributions (side-by-side) of the efficiency of hospital 1 (DMUs 1, 28, 55, and 82), per year (2013-2016), and concerning the scenarios (b) hyper-ellipsoid and (c) hyper-rhombus. Note: GEV stands for Generalized Extreme Value (distribution).}
\label{fig:densities_b_c_DMU1_eff}
\end{figure}

\begin{table}[H]
  \centering
  \caption{Parameters of the GEV distributions describing the density of efficiency of hospital 1 (DMUs 1, 28, 55, and 82), for scenarios (b) hyper-ellipsoid and (c) hyper-rhombus.}
  \label{tab:param_GEV}
    \begin{tabular}{ccrrrrrrrr}
    \hline
          &       & \multicolumn{4}{c}{Scenario (b)} & \multicolumn{4}{c}{Scenario (c)} \\
          &       & 2013  & 2014  & 2015  & 2016  & 2013  & 2014  & 2015  & 2016 \\
          \hline
    \multicolumn{1}{r}{Location} & \multicolumn{1}{r}{$\rho$} & 0.5909 & 0.5453 & 0.5330 & 0.5689 & 0.5869 & 0.5439 & 0.5310 & 0.5679 \\
    \multicolumn{1}{r}{Scale} & \multicolumn{1}{r}{$\sigma$} & 0.0085 & 0.0106 & 0.0350 & 0.0102 & 0.0118 & 0.0102 & 0.0093 & 0.0102 \\
    \multicolumn{1}{r}{Shape} & \multicolumn{1}{r}{$\psi$} & -0.3344 & -0.3564 & -0.6350 & -0.3395 & -0.4860 & -0.4339 & -0.4030 & -0.4468 \\
    \multicolumn{2}{c}{$g_1 = \Gamma(1-\psi)$} & 0.8929 & 0.8905 & 0.8979 & 0.8923 & 0.8859 & 0.8859 & 0.8871 & 0.8857 \\
    \multicolumn{2}{c}{$g_2 = \Gamma(1-2\psi)$} & 0.9031 & 0.9111 & 1.1462 & 0.9048 & 0.9885 & 0.9511 & 0.9330 & 0.9596 \\
    \hline
    \multicolumn{2}{c}{$\mathbb{E}(\Theta^k)$} & 0.5936 & 0.5486 & 0.5386 & 0.5721 & 0.5897 & 0.5466 & 0.5336 & 0.5705 \\
    \multicolumn{2}{c}{$Var(\Theta^k)$} & 0.0001 & 0.0001 & 0.0010 & 0.0001 & 0.0001 & 0.0001 & 0.0001 & 0.0001 \\
    \multicolumn{2}{c}{$\sqrt{Var(\Theta^k)}$} & 0.0083 & 0.0102 & 0.0321 & 0.0099 & 0.0110 & 0.0096 & 0.0088 & 0.0096 \\
    \hline
    \end{tabular}
\end{table}%




\section{Concluding remarks and future research directions}\label{sec:Concluding remarks}
\noindent Our paper proposes the application of a \gls{HR} procedure to \gls{DEA} to introduce a stochastic nature into the latter. The deterministic nature of \gls{DEA} is commonly pointed out as one of its main shortcomings. This way, we introduce the possibility of making statistical inference with efficiency estimates at the same time that we account for \gls{IKD}, which is a problem plaguing most of databases. The proposed procedure generalizes the interval \gls{DEA}, which provides the broadest confidence intervals for efficiency when \gls{IKD} is modeled using intervals. The integrated approach \gls{HR}+\gls{DEA} has a superior performance comparatively with some widely spread alternatives such as regressions. Unless we know precisely the mathematical function underlying the production process, using a predefined function to fit and estimate these \gls{IKD} (typically gaps) is only a matter of \textit{shooting in the dark}. Besides, the stochastic nature of efficiency scores obtained through the integrated approach is its greatest advantage. The model is simple to implement and allows the researcher/expert/decision-maker to bound imperfect knowledge directly on data, making her/his job easier in most of the empirical situations.

In the near future we expect to test our integrated approach for robustness in terms of the quantity (or rate) of observations to account for \gls{IKD}, the size or dimensions of each set (as defined by the parameters $w$), and other shapes for sets $\Lambda$ (namely, the possibility of including non-convex sets).

\singlespace
\section*{Acknowledgements}
\addcontentsline{toc}{section}{\numberline{}Acknowledgements}
\noindent  This work was financially supported by the hSNS FCT -- Research Project  (02/SAICT/2017/30546). Jos\'{e} Rui Figueira also acknowledges the support from the FCT grant SFRH/BSAB/139892/2018 under POCH Program during his stay at the Department of Mathematics of the University of Wuppertal, Germany.



\addcontentsline{toc}{section}{\numberline{}References}
\bibliographystyle{model2-names}
\bibliography{References}

\begin{figure}[htbp]
\centering
	\subfloat[$O_1=O_2 = +\infty$]{\includegraphics[width = .33\columnwidth]{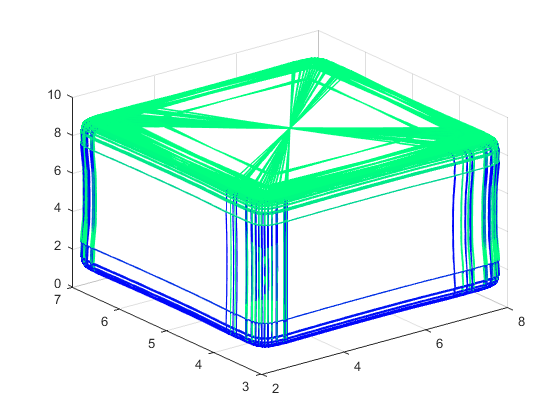}} 
	\subfloat[$O_1=O_2=2$]{\includegraphics[width = .33\columnwidth]{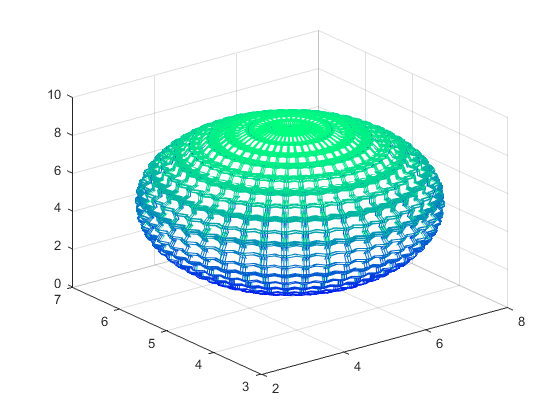}}
    \subfloat[$O_1=O_2=1 $]{\includegraphics[width = .33\columnwidth]{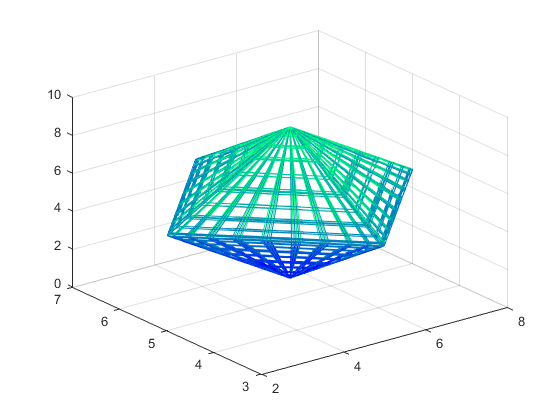}}
\caption{Three different ways of modeling IKD using super-ellipsoids centered in $(5,5,5)$ and with semi-axes $(3,2,4)$: $\left(\left|\frac{x_1 -5}{3}\right|^{O_1}+\left|\frac{x_2 -5}{2}\right|^{O_1}\right)^{O_2/O_1} +\left|\frac{u_1 -5}{4}\right|^{O_2} =1$ (with 5,000 iterations).}
\label{fig:superEllips2}
\end{figure}
\begin{landscape}
\begin{table}[htbp]
  \centering
  \caption{Efficiency estimation using the Hit \& Run approach and the scenario (b). Results for 2015 and 2016.\label{tab:scenario_b_2015_2016}}
    \begin{tabular}{lrrrrrrcrrrrrrc}
    \hline
         & \multicolumn{7}{c}{2015}                              & \multicolumn{7}{c}{2016} \\
          \hline
    Hospital & \multicolumn{1}{l}{$ERII_0^k$} & \multicolumn{1}{l}{$\mathbb{E}(D^k)$} & \multicolumn{1}{l}{$LB_{95\%}$} & \multicolumn{1}{l}{$UB_{95\%}$} & \multicolumn{1}{l}{$\sqrt{Var(D^k)}$} & \multicolumn{1}{l}{$\mathbb{E}(\Theta^k)$} & \multicolumn{1}{c}{Category} & \multicolumn{1}{l}{$ERII_0^k$} & \multicolumn{1}{l}{$\mathbb{E}(D^k)$} & \multicolumn{1}{l}{$LB_{95\%}$} & \multicolumn{1}{l}{$UB_{95\%}$} & \multicolumn{1}{l}{$\sqrt{Var(D^k)}$} & \multicolumn{1}{l}{$\mathbb{E}(\Theta^k)$} & \multicolumn{1}{c}{Category} \\
    \hline
     \multicolumn{1}{r}{1} & 0.00  & 0.3005 & 0.2623 & 0.3531 & 0.0264 & 0.5385 & \multicolumn{1}{l}{$C_4$} & 0.00  & 0.2722 & 0.2572 & 0.2880 & 0.0081 & 0.5721 & \multicolumn{1}{l}{$C_4$} \\
    \multicolumn{1}{r}{2} & 0.00  & 0.3208 & 0.2517 & 0.4134 & 0.0476 & 0.5162 & \multicolumn{1}{l}{$C_4$} & 0.00  & 0.3054 & 0.2680 & 0.3354 & 0.0169 & 0.5324 & \multicolumn{1}{l}{$C_4$} \\
    \multicolumn{1}{r}{3} & 0.00  & 0.1538 & 0.1487 & 0.1590 & 0.0029 & 0.7335 & \multicolumn{1}{l}{$C_4$} & 0.00  & 0.1952 & 0.1577 & 0.2328 & 0.0209 & 0.6738 & \multicolumn{1}{l}{$C_4$} \\
    \multicolumn{1}{r}{4} & 0.00  & 0.0727 & 0.0604 & 0.0839 & 0.0064 & 0.8645 & \multicolumn{1}{l}{$C_4$} & 0.00  & 0.0725 & 0.0614 & 0.0825 & 0.0057 & 0.8649 & \multicolumn{1}{l}{$C_4$} \\
    \multicolumn{1}{r}{5} & 0.00  & 0.2612 & 0.2534 & 0.2691 & 0.0043 & 0.5858 & \multicolumn{1}{l}{$C_4$} & 0.00  & 0.2594 & 0.2517 & 0.2673 & 0.0042 & 0.5881 & \multicolumn{1}{l}{$C_4$} \\
    \multicolumn{1}{r}{6} & 0.00  & 0.0487 & 0.0336 & 0.0657 & 0.0086 & 0.9072 & \multicolumn{1}{l}{$C_4$} & 0.00  & 0.0441 & 0.0310 & 0.0584 & 0.0075 & 0.9157 & \multicolumn{1}{l}{$C_4$} \\
    \multicolumn{1}{r}{7} & 0.00  & 0.2464 & 0.2350 & 0.2585 & 0.0064 & 0.6046 & \multicolumn{1}{l}{$C_4$} & 0.00  & 0.2845 & 0.2422 & 0.3297 & 0.0250 & 0.5576 & \multicolumn{1}{l}{$C_4$} \\
    \multicolumn{1}{r}{8} & 0.00  & 0.2253 & 0.2157 & 0.2364 & 0.0056 & 0.6322 & \multicolumn{1}{l}{$C_4$} & 0.00  & 0.2528 & 0.2372 & 0.2680 & 0.0090 & 0.5965 & \multicolumn{1}{l}{$C_4$} \\
    \multicolumn{1}{r}{9} & 0.00  & 0.1604 & 0.1187 & 0.2013 & 0.0232 & 0.7242 & \multicolumn{1}{l}{$C_4$} & 0.00  & 0.1855 & 0.1773 & 0.1938 & 0.0045 & 0.6871 & \multicolumn{1}{l}{$C_4$} \\
    \multicolumn{1}{r}{10} & 0.62  & 0.0501 & 0.0000 & 0.1982 & 0.0731 & 0.9138 & \multicolumn{1}{l}{$C_3$} & 0.62  & 0.2409 & 0.2297 & 0.2521 & 0.0059 & 0.6118 & \multicolumn{1}{l}{$C_4$} \\
    \multicolumn{1}{r}{11} & 0.00  & 0.0641 & 0.0510 & 0.0773 & 0.0080 & 0.8796 & \multicolumn{1}{l}{$C_4$} & 0.00  & 0.0671 & 0.0535 & 0.0808 & 0.0083 & 0.8743 & \multicolumn{1}{l}{$C_4$} \\
    \multicolumn{1}{r}{12} & 0.00  & 0.0668 & 0.0627 & 0.0712 & 0.0023 & 0.8747 & \multicolumn{1}{l}{$C_4$} & 0.00  & 0.0536 & 0.0521 & 0.0550 & 0.0008 & 0.8982 & \multicolumn{1}{l}{$C_4$} \\
    \multicolumn{1}{r}{13} & 0.00  & 0.0180 & 0.0172 & 0.0188 & 0.0004 & 0.9646 & \multicolumn{1}{l}{$C_4$} & 0.00  & 0.0000 & 0.0000 & 0.0000 & 0.0000 & 1.0000 & \multicolumn{1}{l}{$C_1$} \\
    \multicolumn{1}{r}{14} & 0.00  & 0.2037 & 0.1899 & 0.2173 & 0.0076 & 0.6616 & \multicolumn{1}{l}{$C_4$} & 0.00  & 0.1995 & 0.1885 & 0.2094 & 0.0057 & 0.6673 & \multicolumn{1}{l}{$C_4$} \\
    \multicolumn{1}{r}{15} & 0.00  & 0.2473 & 0.2349 & 0.2595 & 0.0065 & 0.6035 & \multicolumn{1}{l}{$C_4$} & 0.00  & 0.2462 & 0.2341 & 0.2584 & 0.0064 & 0.6049 & \multicolumn{1}{l}{$C_4$} \\
    \multicolumn{1}{r}{16} & 0.54  & 0.0129 & 0.0000 & 0.0613 & 0.0188 & 0.9753 & \multicolumn{1}{l}{$C_3$} & 0.54  & 0.0162 & 0.0000 & 0.0702 & 0.0221 & 0.9691 & \multicolumn{1}{l}{$C_3$} \\
    \multicolumn{1}{r}{17} & 0.00  & 0.0571 & 0.0550 & 0.0593 & 0.0012 & 0.8919 & \multicolumn{1}{l}{$C_4$} & 0.00  & 0.0748 & 0.0719 & 0.0776 & 0.0015 & 0.8609 & \multicolumn{1}{l}{$C_4$} \\
    \multicolumn{1}{r}{18} & 0.00  & 0.1670 & 0.1618 & 0.1723 & 0.0028 & 0.7139 & \multicolumn{1}{l}{$C_4$} & 0.00  & 0.1890 & 0.1836 & 0.1946 & 0.0030 & 0.6821 & \multicolumn{1}{l}{$C_4$} \\
    \multicolumn{1}{r}{19} & 0.00  & 0.2018 & 0.1958 & 0.2082 & 0.0034 & 0.6641 & \multicolumn{1}{l}{$C_4$} & 0.00  & 0.1885 & 0.1827 & 0.1944 & 0.0032 & 0.6829 & \multicolumn{1}{l}{$C_4$} \\
    \multicolumn{1}{r}{20} & 0.00  & 0.1192 & 0.1000 & 0.1387 & 0.0114 & 0.7872 & \multicolumn{1}{l}{$C_4$} & 0.00  & 0.1353 & 0.1153 & 0.1550 & 0.0115 & 0.7619 & \multicolumn{1}{l}{$C_4$} \\
    \multicolumn{1}{r}{21} & 0.00  & 0.1558 & 0.1491 & 0.1628 & 0.0037 & 0.7304 & \multicolumn{1}{l}{$C_4$} & 0.00  & 0.1701 & 0.1628 & 0.1777 & 0.0041 & 0.7093 & \multicolumn{1}{l}{$C_4$} \\
    \multicolumn{1}{r}{22} & 0.00  & 0.0214 & 0.0146 & 0.0541 & 0.0102 & 0.9582 & \multicolumn{1}{l}{$C_4$} & 0.00  & 0.0470 & 0.0349 & 0.0895 & 0.0147 & 0.9107 & \multicolumn{1}{l}{$C_4$} \\
    \multicolumn{1}{r}{23} & 0.00  & 0.2786 & 0.2440 & 0.3070 & 0.0175 & 0.5645 & \multicolumn{1}{l}{$C_4$} & 0.00  & 0.2753 & 0.2056 & 0.3597 & 0.0476 & 0.5704 & \multicolumn{1}{l}{$C_4$} \\
    \multicolumn{1}{r}{24} & 0.00  & 0.1968 & 0.1858 & 0.2076 & 0.0058 & 0.6712 & \multicolumn{1}{l}{$C_4$} & 0.00  & 0.2074 & 0.1594 & 0.2601 & 0.0292 & 0.6574 & \multicolumn{1}{l}{$C_4$} \\
    \multicolumn{1}{r}{25} & 0.00  & 0.0787 & 0.0730 & 0.0851 & 0.0033 & 0.8542 & \multicolumn{1}{l}{$C_4$} & 0.00  & 0.0735 & 0.0697 & 0.0778 & 0.0022 & 0.8631 & \multicolumn{1}{l}{$C_4$} \\
    \multicolumn{1}{r}{26} & 0.00  & 0.1958 & 0.1869 & 0.2054 & 0.0050 & 0.6725 & \multicolumn{1}{l}{$C_4$} & 0.00  & 0.2136 & 0.2058 & 0.2221 & 0.0044 & 0.6480 & \multicolumn{1}{l}{$C_4$} \\
    \multicolumn{1}{r}{27} & 0.00  & 0.0183 & 0.0089 & 0.0989 & 0.0208 & 0.9649 & \multicolumn{1}{l}{$C_4$} & 0.00  & 0.0000 & 0.0000 & 0.0000 & 0.0000 & 1.0000 & \multicolumn{1}{l}{$C_1$} \\
    \hline
    Mean  & 0.0431 & 0.1461 & 0.1300 & 0.1720 & 0.0123 & 0.7575 &       & 0.0431 & 0.1581 & 0.1420 & 0.1774 & 0.0101 & 0.7393 &  \\
    Std.Dev. & 0.1529 & 0.0943 & 0.0893 & 0.0981 & 0.0156 & 0.1442 &       & 0.1529 & 0.0955 & 0.0882 & 0.1029 & 0.0107 & 0.1470 &  \\
    CV    & 354\% & 65\%  & 69\%  & 57\%  & 126\% & 19\%  &       & 354\% & 60\%  & 62\%  & 58\%  & 106\% & 20\%  &  \\
    \hline
    \end{tabular}
\end{table}%
\end{landscape}

\end{document}